\documentclass{vldb}
\usepackage{amsmath}    
\usepackage{mathtools}	
\usepackage{graphicx}   
\usepackage{verbatim}   
\usepackage[usenames,dvipsnames,svgnames,table]{xcolor}     
\usepackage{bm}                     	
\usepackage{tabularx}		
\usepackage{aliascnt}  		
\usepackage{balance}
\usepackage{tikz}
\usetikzlibrary{shapes,arrows,calc,positioning,matrix,decorations,patterns}
\usepackage{pgfplots}
\pgfplotsset{compat=1.5}
\usepackage{caption}
\usepackage{subcaption}
\usepackage{microtype}
\usepackage{paralist}
\usepackage{multirow}
\usepackage[inline]{enumitem}
\usepackage{afterpage}
\usepackage{dblfloatfix}
\usepackage{fixltx2e}

\newcommand{\eat}[1]{}					

\newcommand{\smallsection}[1]{\vspace{2mm}\noindent\textbf{#1.}}	

\newtheorem{theorem}{Theorem}          	
\newaliascnt{lemma}{theorem}				
\aliascntresetthe{lemma}  					
\newaliascnt{conjecture}{theorem}			
\aliascntresetthe{conjecture}  				
\newaliascnt{remark}{theorem}				
              
\aliascntresetthe{remark}  					
\newaliascnt{corollary}{theorem}			
\newtheorem{corollary}[corollary]{Corollary}      
\aliascntresetthe{corollary}  				
\newaliascnt{definition}{theorem}			
\newtheorem{definition}[definition]{Definition}    
\aliascntresetthe{definition}  				
\newaliascnt{proposition}{theorem}			
\newtheorem{proposition}[proposition]{Proposition}  
\aliascntresetthe{proposition}  				
\newaliascnt{example}{theorem}			
\newtheorem{example}[example]{Example}  	
\aliascntresetthe{example}  				
\newaliascnt{problem}{theorem}			
\newtheorem{problem}[problem]{Problem}  	
\aliascntresetthe{problem}  				

\usepackage[algo2e, vlined, ruled]{algorithm2e}	
\DontPrintSemicolon     		
\SetNlSty{}{}{}					
\SetAlgoCaptionSeparator{.}		
\SetAlgoSkip{smallskip}			
\let\oldnl\nl
\newcommand{\nonl}{\renewcommand{\nl}{\let\nl\oldnl}}

\SetKwFor{ForAll}{forall}{do}{endfch}	
\LinesNumbered

\graphicspath{{img/}}	

\renewenvironment{thebibliography}[1]{%
    \begin{oldthebibliography}{#1}%
        \setlength{\parskip}{0.1ex}
        \setlength{\itemsep}{0.0ex}%
}%
{%
\end{oldthebibliography}%
}

\definecolor{green1}{RGB}{0, 102, 0}
\definecolor{green2}{RGB}{0, 204, 0}
\definecolor{red1}{RGB}{255, 99, 71}
\definecolor{blue1}{RGB}{0, 0, 204}

\usepackage{catoptions}
\makeatletter

\def\Autoref#1{%
  \begingroup
  \edef\reserved@a{\cpttrimspaces{#1}}%
  \ifcsndefTF{r@#1}{%
    \xaftercsname{\expandafter\testreftype\@fourthoffive}
      {r@\reserved@a}.\\{#1}%
  }{%
    \ref{#1}%
  }%
  \endgroup
}
\def\testreftype#1.#2\\#3{%
  \ifcsndefTF{#1autorefname}{%
    \def\reserved@a##1##2\@nil{%
      \uppercase{\def\ref@name{##1}}%
      \csn@edef{#1autorefname}{\ref@name##2}%
      \autoref{#3}%
    }%
    \reserved@a#1\@nil
  }{%
    \autoref{#3}%
  }%
}
\makeatother

\usepackage{hyperref}   

\begin{document}

\newcommand{\anja}[1]{\textcolor{ForestGreen}{ANJA: #1}}
\newcommand{\besa}[1]{\textcolor{orange}{BESA: #1}}
\newcommand{\tim}[1]{\textcolor{blue}{TIM: #1}}

\title{Fault-Tolerant Entity Resolution with the Crowd}

\numberofauthors{1}
\author{
\centering
\begin{tabular}{ccc}
Anja Gruenheid, Besmira Nushi & Wolfgang Gatterbauer & Tim Kraska\\
Donald Kossmann & {\large Tepper School of Business} & {\large Brown University} \\
{\large Systems Group} & {\large Carnegie Mellon University} & tim\_kraska@brown.edu\\
{\large Dep. of Computer Science} & gatt@andrew.cmu.edu & \\
{\large ETH Zurich} & & \\
\{agruen, nushib, donaldk\}@inf.ethz.ch &  & \\
\end{tabular}
}

\maketitle

\begin{abstract}
In recent years, crowdsourcing is increasingly applied as a means to enhance data quality.
Although the crowd generates insightful information especially for complex problems such as entity resolution (ER), the output quality of crowd workers is often noisy.
That is, workers may unintentionally generate false or contradicting data even for simple tasks.
The challenge that we address in this paper is how to minimize the cost for task requesters while maximizing ER result quality under the assumption of unreliable input from the crowd.
For that purpose, we first establish how to deduce a consistent ER solution from noisy worker answers as part of the {\it data interpretation} problem.
We then focus on the {\it next-crowdsource} problem which is to find the next task that maximizes the information gain of the ER result for the minimal additional cost.
We compare our robust data interpretation strategies to alternative state-of-the-art approaches that do not incorporate the notion of {\em fault-tolerance}, i.e.,~the robustness to noise.
In our experimental evaluation we show that our approaches yield a quality improvement of at least 20\% for two real-world datasets.
Furthermore, we examine task-to-worker assignment strategies as well as task parallelization techniques in terms of their cost and quality trade-offs in this paper.
Based on both synthetic and crowdsourced datasets, we then draw conclusions on how to minimize cost while maintaining high quality ER results.
\end{abstract}



\section{Introduction}\label{sec:introduction}

Data cleaning and data integration are integral techniques for analytical and personalized data systems.
Many efficient automated mechanisms addressing both of these problems have been integrated into large-scale systems over the last decades.
Recently, several studies have shown that crowdsourcing can produce higher quality solutions for a subset of the data integration tasks \cite{DBLP:conf/nips/GomesWKP11,DBLP:conf/sigmod/WangLKFF13}.
For complex problems such as entity resolution (ER) or picture classification crowdsourcing has been established as an alternative to automated techniques.
In fact, approaches that prune the search space with automated ER mechanisms and then enhance data quality through crowdsourcing are common for a large number of high profile ER systems such as the Google Knowledge Graph \cite{singhal2012introducing} or the Facebook Entities Graph \cite{fbentitiesgraph}.
Even though the overall result quality of ER solutions generally benefits from human input, it can also be observed that crowd workers may make mistakes when executing tasks.
These mistakes may be the result of carelessness, ambiguities in the task description, or even malicious behavior.
More specifically, it is common to have crowd error rates as high as 30\% \cite{Ipeirotis:2010:QMA:1837885.1837906} on well-established crowdsourcing platforms.
\Autoref{fig:animals} shows an example for a task which may mislead humans and how these mistakes could be avoided intuitively. 

\begin{example}[Animal Classification]\label{ex:animals}
Crowd workers are given the task to determine which animals belong to the same breed, i.e.,~$r_1$ and $r_2$ are seals, $r_3$ and $r_4$ show sea lions.
However, we observe that distinguishing the full-grown animals (records $r_2$ and $r_4$) is easier than telling the baby animals (records $r_1$ and $r_3$) apart because they have similar appearance characteristics.
\end{example}

\begin{figure}
	\centering
	\includegraphics[width=0.65\columnwidth]{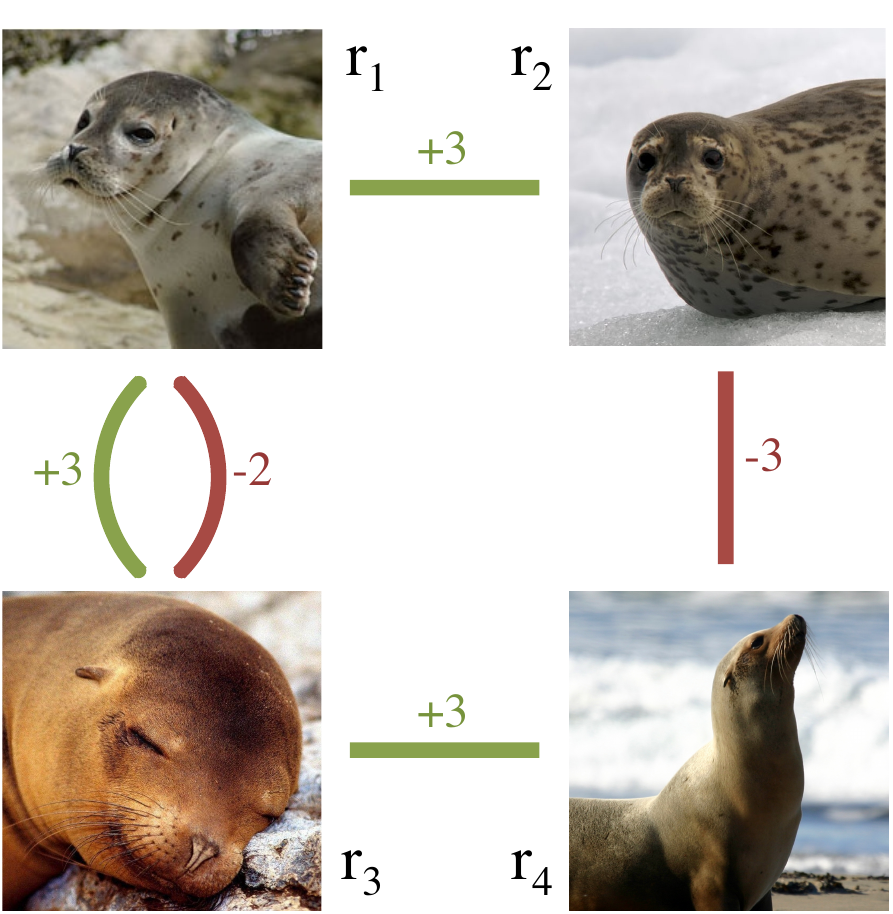}
	\caption{Animal classification problem.}\label{fig:animals}
\end{figure}

In this paper, we study the problem of crowdsourced entity resolution with potentially erroneous input by crowd workers which we refer to as {\it fault-tolerant} entity resolution.
We choose fault-tolerant as naming convention because crowd workers make faulty decisions that have to be tolerated by the ER engine.
The terminology here indicates that we take all available information from the crowd into consideration when making robust decisions about the ER solution.

As worker input is usually not for free, our goal is to minimize the monetary cost while maximizing the overall ER result quality. 
Our solution consists of two components: 
First, we address {\em data interpretation}, i.e.,~how the answers provided by crowd workers lead to an ER solution, and discuss how it can be efficiently implemented for ER computation with potentially erroneous crowd answers. 
Second, we focus on minimizing the cost that crowdsourcing incurs which is also called the {\em next-crowdsource} problem.

\smallsection{Data Interpretation}
Prior work in crowdsourced entity resolution \cite{DBLP:journals/pvldb/VesdapuntBD14,DBLP:conf/sigmod/WangLKFF13,DBLP:journals/pvldb/WhangLG13} has focused on handling data interpretation and crowd worker quality as separate problems.
For instance, this line of work proposed to use qualification tasks \cite{Kazai:2011:WTP:2063576.2063860} to filter out malicious or low-performing workers or to use replicated tasks with quorum votes (e.g.,~majority of answers) to determine the correct answers in ambiguous cases \cite{Nowak:2010:RAV:1743384.1743478}. 
If data quality is thus ensured, the ER algorithm can be designed under the assumption that there are no conflicts in the dataset.
In contrast, we argue in our work that there are numerous sources for erroneous information from the crowd that a good ER algorithm needs to interpret correctly.
Consider the following example:
If workers misclassify the baby sea lion and seal ($r_1$ and $r_3$) with three positive and two negative votes and correctly identify the sea lion ($r_1$ and $r_2$) and seal pairs ($r_3$ and $r_4$) then all of these animals would be wrongly classified as the same species.
Under the assumption of a majority-based decision scheme, even close decisions such as [$r_1$,~$r_3$] are not questioned although the indecision of the crowd clearly indicates that this relationship is uncertain and should be explored further.

To overcome these mistakes, we propose a graph-based path ER technique to identify and appropriately interpret noisy data.
Our technique considers both positive and negative (indirect) crowd answers between two records and provides provably better quality than majority-based approaches that only value the dominant decision.
In addition, we discuss how to integrate available worker input efficiently into the decision making process, allowing for a consistent ER solution at any point in time.

\smallsection{Next-Crowdsource Problem}
In addition to interpreting noisy data, we look at ways to minimize the cost that ER with erroneous crowd information incurs.
In that context, we focus on (a) crowd task ordering and (b) task parallelization strategies.
Generally, crowdsourced ER is executed on platforms such as Amazon Mechanical Turk \cite{amt} which allow task requesters to employ crowd workers for monetary compensation.
These task providers therefore need to devise effective task assignment strategies that minimize the overall monetary cost by maximizing the information gain per cost unit.
To address the challenge of task ordering, we develop three different ordering strategies that can be employed in the context of crowdsourced ER and examine them with synthetic and real-world datasets.
We furthermore explore task parallelization in the second part of this work.
It is a mechanism that tries to minimize the end-to-end runtime of crowdsourced ER by publishing multiple tasks at the same time on the platform.
In our discussion of parallelization strategies, we show the trade-off between the runtime acceleration of crowdsourced ER and its monetary cost and output quality.
We make the following contributions:
\begin{itemize}
  \item {\bf Fault-tolerant entity resolution}. We formally introduce fault-tolerant decision functions for ER that decide whether two records belong to the same entity or not. To handle unreliable information, we devise a path-based graph interpretation mechanism and define a clustering algorithm that computes the ER solution based on noisy pair-wise decisions to address the data interpretation problem.
  \item {\bf Cost and quality optimization}. Given that the crowd is expensive to employ on a large scale, we discuss task ordering and parallelization strategies and the impact of these mechanisms on data quality and incurred cost. These mechanisms are essential to solve the next-crowdsource problem.
\end{itemize}

This paper is structured as follows:
We first give an overview of the problems that this work addresses and outline how our work can be classified in the context of already existing work in \Autoref{sec:problem}.
We then discuss our solution to the data interpretation problem (\Autoref{sec:faulttolerant} and \Autoref{sec:er}) and the next-crowdsource problem (\Autoref{sec:next}).
Our solutions are evaluated through experiments with datasets obtained from real-world as well as synthetic crowdsourcing setups in \Autoref{sec:experiments}.

\section{Overview}\label{sec:problem}
In this work, we discuss mechanisms to enable entity resolution with imperfect answers from crowd workers.
In order to accurately capture the information provided by them, we need an efficient data structure that allows us to encode both their positive and negative signals. 
For this purpose, we introduce the notion of a {\em votes graph} that stores this information and will later allow us to efficiently interpret the crowd worker's answers.
An example for a votes graph can be seen in \Autoref{fig:animals}.

\begin{definition}[Votes Graph]
A votes graph $G = (R,E)$ is a weighted undirected graph that consists of records $R$ as its nodes and a set of edges $E$ which determine the direct relationship between any records $r_i$ and $r_j$, with $r_i$,$r_j$ $\in$ $R$.
For each record pair [$r_i$,$r_j$] there exists at most one positive edge $p_{ij}$ and one negative edge $n_{ij}$.
They correspond to the positive and negative crowd answers for this record pair.
That is, while $p_{ij}$ corresponds to the total number of votes that say that $r_i$ and $r_j$ belong to the same real-world entity, $n_{ij}$ are the votes against it.
\end{definition}
For simplicity, we assume that each vote from each data source (i.e.,~a specific worker) has the same weight.
However, note that weighting schemes for (un-)reliable crowd workers can be easily integrated into the votes graph:
For example if the system assumes a crowd worker to have complete knowledge, it can transform that worker's votes to the maximal edge weight.
Other transformations can be computed analogously.

\subsection{Framework}
There are two steps that are integral to automated ER systems on a (votes) graph $G$ which are essential to understanding the general framework for crowdsourced ER.
First, they establish all pair-wise similarities for any two records $r_i$ and $r_j$.
Second, they determine a clustering $C$ that optimizes the record-to-entity assignment according to some objective function, for example transitive closure or penalty optimization.
We will go into detail on clustering alternatives later in this work, please refer to \Autoref{sec:er} for details.
This well-established pipeline for automated ER is similar but not equivalent to the pipeline for crowdsourced ER.
Do understand the core differences, remember that the comparison of records in a crowdsourced environment raises two questions:
\begin{enumerate} 
  \item Given that each additional edge incurs a monetary cost, which edges do I really need to know about?
  \item And if $G$ is incomplete and there exist no direct votes between $r_i$ and $r_j$, how can we estimate the similarity between them?
\end{enumerate}
The first question is raised due to budgetary limitations that are inherent to crowdsourcing applications.
Through these limitations and the choice of observed edges, other record pairs may only be known indirectly or not at all.
To predict an accurate record-to-entity assignment, it is therefore necessary to estimate their relationship.
ER solutions that are designed for the crowd therefore focus on two slightly different core problems than traditional ER:
The first problem is how to understand, interpret, and enrich the data that has been retrieved from the crowd which we refer to in the following as the {\em data interpretation} problem.
The second problem is to determine the {\em next-crowdsource} order in which record pair information is requested from the crowd.

To address both of these problems, crowdsourced ER algorithms \cite{DBLP:journals/pvldb/VesdapuntBD14,DBLP:conf/sigmod/WangLKFF13} generally follow the same incremental process outlined in Algorithm~\ref{alg:generalER}.
Initially, they generate a set of record pairs ({\em candidate pairs}) $[r_i,r_j]$.
This could be a complete set of record pairs or subset of those (for example using automatic similarity comparison to remove unlikely record pairs).
These pairs are then sorted with regard to a predefined priority metric and added to a priority queue $Q$ (Line~\ref{alg:gen:cluster:add}).
An example for such a priority queue is to order them according to their similarity in an automatic similarity computation.
Iteratively, the top record pair is now retrieved and published as a task on a platform like Amazon Mechanical Turk to obtain information whether $r_i$ and $r_j$ in fact belong to the same entity (Line~\ref{alg:gen:cluster:iterate}).
Whenever a crowd worker responds, the answer is integrated into the votes graph $G$.
Based on the new information in $G$, a clustering algorithm then determines the current entity resolution solution $C$ (Line~\ref{alg:gen:cluster:datainterp}).
For example a simple clustering mechanism would be to merge all records $r_i$ and $r_j$ into the same real-world entity if $p_{ij} > n_{ij}$.
Finally, this new information from the crowd may influence other record pairs, which can lead to an adjustment of the priority queue for record pairs (Line~\ref{alg:gen:cluster:next}).
For example, if $r_i$ and $r_j$ as well as $r_j$ and $r_k$ are assigned to the same cluster, then asking for record pair [$r_i$,$r_k$] is superfluous if the algorithm exploits transitivity.

\begin{algorithm2e}[t]
\SetKwInput{KwInput}{Input}
\SetKwInput{KwOutput}{Output}
\small
\output{Clustering $C$}
\BlankLine
{
$Q, C \leftarrow \emptyset$\; 
\hspace{.02in}\textcolor{gray}{// add elements to the queue}\\
\lForEach{$[r_i,r_j] \in R^2; i\neq j$} {
	$Q$.priorityAdd$([r_i,r_j])$\;\label{alg:gen:cluster:add}
}
\hspace{.02in}\textcolor{gray}{// iteratively crowdsource and adjust clustering}\\
\ForEach{$[r_i,r_j] \in Q$\label{alg:gen:cluster:iterate}} {
	$v_{ij} \leftarrow$ crowdsource$([r_i,r_j])$\;
	\hspace{.02in}\textcolor{Mahogany}{// data interpretation problem}\\
	\textcolor{Mahogany}{C.update($v_{ij}$)\;}\label{alg:gen:cluster:datainterp}
	\hspace{.02in}\textcolor{CornflowerBlue}{// next-crowdsource problem}\\
	\textcolor{CornflowerBlue}{Q.adjust($C$)\;}\label{alg:gen:cluster:next}
}
}
{\bf return} $C$
\caption{General ER algorithm for crowdsourced on-demand input.}\label{alg:generalER}
\end{algorithm2e}

\subsubsection{Problem Definition}\label{sec:sub:problem}
To describe the data interpretation and next-crowdsource problems formally, we first define what an optimal solution entails in crowdsourced ER.
For the data interpretation problem, the goal is to find a clustering $C^*$ that represents the correct entity resolution solution, i.e.,~the record-to-entity mapping is equivalent to some ground truth.
Finding such a clustering is straightforward if all pair-wise similarity estimates are correct.
In other words, if there exists an oracle that knows whether any two records $r_i, r_j \in R$ belong to the same real-world entity then finding $C^*$ with any of the well-established ER algorithms is trivial.
The reason is that there are no contradictions in the edge set (i.e.,~there are no $i,j,k$ s.t.~$r_i = r_j$ and $r_i = r_k$ but $r_j \neq r_k$).
That means that there exists one unique clustering $C^*$ that is consistent with the specified pair-wise relations.

Under the assumption of incomplete or incorrect edges, i.e.,~erroneous votes from the crowd workers, finding $C^*$ becomes an optimization problem.
Specifically, we want to find the clustering $C$ that is either equivalent to $C^*$ or the best approximation thereof based on the current state of the votes graph $G$.
To compute $C$, our algorithms estimate the distance to the optimal solution through a distance measure $d$ which represents the similarity of $C$ to $C^*$.
It is essentially an objective function that estimates the correctness of a solution based on the clustering mechanics that define $C^*$.
Examples for $d$ include minimizing the number of negative similarity scores within the same and positive scores across clusters, or to maximize the number of positive edges within a cluster.
The choice of optimization metric is dependent on the applied ER algorithm.
However, assume in the following that the goal for all candidate mechanisms is to minimize $d(C,C^*)$ where $d(C,C^*)$ is the distance of $C$ to $C^*$ for simplicity.
We can then formulate the data interpretation problem as follows.
\begin{problem}[Data Interpretation Problem]
Given are all record pairs $[r_i,r_j] \in G$ and a metric $d(C,C^*)$ that assigns a distance for any given clustering $C$ from ground truth $C^*$. 
The data interpretation problem is to find the best clustering $C$ such that $d(C,C^*) < d(C',C^*)$ for any alternative clustering $C'$ of $G$.
\end{problem}
The second problem of ER with a crowd is that crowdsourcing platforms incur cost for the task requester.
The goal when issuing tasks, i.e.,~obtaining more information for specified edges in $G$, therefore becomes finding those tasks that provide maximal information for minimal cost.
Minimizing the task space through intelligent vote requests is part of the next-crowdsource problem.
In that context, task ordering is essential for efficient crowdsourced ER because it is highly correlated to the output quality in addition to the required (monetary) cost.
To understand why, remember that the system needs to estimate pair-wise decisions in the data interpretation problem.
Thus, it provides better solutions if it knows which questions enable fast convergence to the best possible clustering solution.
The next-crowdsource problem can therefore be formalized as follows.
\begin{problem}[Next-Crowdsource Problem]
Given a votes graph $G$ and a distance function $d(C,C^*)$. The next-crowdsource problem is to choose a record pair [$r_i$, $r_j$]$\in G$ for a crowd worker to vote on, such that adding the outcome of the vote to $G$ minimizes the expected distance $d(C',C^*)$ of the updated clustering $C'$.
\end{problem}

\subsubsection{Worker Quality}
The quality of the entity resolution result is based on the quality of the answers the crowd workers provide.
Formally, we define a correct answer for a record pair $[r_i,r_j]$ as the true positive or true negative answers, i.e.,~if $r_i$ and $r_j$ belong to the same entity, a correct answer is `Yes'.
Similarly, a faulty answer for $[r_i,r_j]$ encompasses false positive and false negative answers of crowd workers.
In our experimental evaluation, we show that faulty crowd answers have drastic impact on the ER quality as observed in pair-wise precision and recall (see \Autoref{subsec:crowderror} for details).
Incorrect answers are often the result of incomplete knowledge in a certain domain or lack of attention during task execution.
In fact, we observe an evenly distributed error rate for both false positive and false negative crowd answers across all workers.
For example, we employed 545 different workers to examine our landmarks dataset (see \Autoref{sec:experiments} for details) out of which only 19 workers (i.e.,~3.5\% of all workeres) deviated significantly from the measured average worker quality because of errors in their answers.
Out of these, only 9 did more than 20 tasks and thus had more significant impact on the results.

The algorithms presented in this paper do not differentiate between crowd workers but rather handle every worker equally.
The reasons for that are two-fold.
First, as mentioned above, we observe little variation in worker quality in our real-world experiments.
Second, the same crowd is often not available for the same task more than once.
As a result, available workers may not have a quality profile which the algorithm can fall back on.
For example, in our landmarks dataset, we observed that only 51.3\% of the workers executed more than 20 tasks.
Obviously, it is only possible to build accurate error statistics per worker if these workers provide sufficient sample work.
To compensate for missing error profiles, the general decision strategies that we explore in this work employ robust decision mechanisms that make them independent of worker error statistics.

\subsection{Classification of Crowdsourced ER}\label{subsec:strategies}
Generally, this work discusses classes of crowdsourced ER captured by the framework shown in \Autoref{fig:ERapproaches}.
They can be differentiated by how they address (a) the next-crowdsource problem and (b) the data interpretation problem.
In the next-crowdsource problem, task ordering can process in a {\it monotonic} manner or alternatively already resolved pairs may be reconsidered if evidence points to a mistake in the previous decisions.
This {\bf non-monotonic} task ordering is based on the assumption that input information is unreliable.
Going back to the introductory example (\Autoref{fig:animals}), imagine that a crowd worker first classifies the baby animals to be of the same kind, identifies the relationship of each baby and respective grown-up animal correctly, and then determines the difference in grown-up animals.
To improve quality, it should be possible to re-evaluate the ER solution to correct the initial mistake of the workers.
In our work we evaluate both non-monotonic and monotonic task ordering in our general execution framework (Line \ref{alg:gen:cluster:datainterp} in Algorithm~\ref{alg:generalER}).
The second difference between different classes of crowdsourced ER is that they either leverage {\bf complete} (positive and negative even if contradictory) votes from the crowd or request a {\bf consensus} decision.
As a result, there exist three different strategies in the solution space for crowdsourced ER:
\begin{enumerate*}[label=\alph*]
\item {\bf fault-tolerant exhaustive},
\item {\bf fault-tolerant}, and
\item {\bf consensus-based}
\end{enumerate*}
strategies.
Notice that combining non-monotonic and consensus-based mechanisms is not possible.
The reason is that consensus mechanisms can never lead to contradictions in the ER solution, thus violating the non-monotonicity property.

In the following, we examine how these categories of crowdsourced differ from each other and show the most prominent examples of existing work in either category.
A more general overview of related work encompassing topics that are also outside of crowdsourced ER can be found in \Autoref{sec:related}.

\begin{figure}[t]
\centering
\includegraphics[width=\columnwidth]{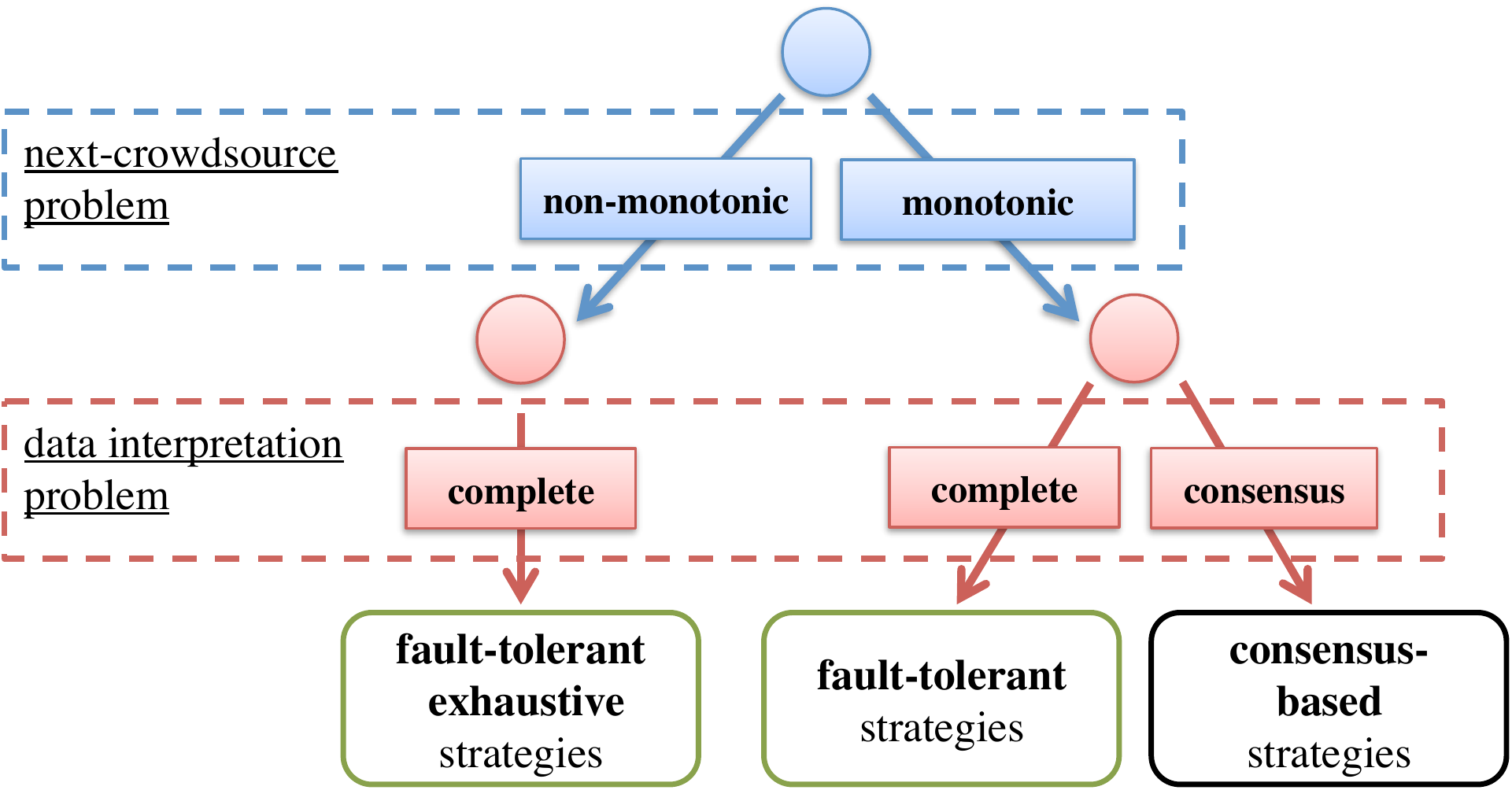}
\caption{Overview of different ER strategies.}
\label{fig:ERapproaches}
\end{figure}

\smallsection{Consensus-Based Entity Resolution}
Entity resolution strategies that are based on consensus decision usually allocate a fixed repetition budget for each crowdsourcing task.
The crowd worker's answers are then aggregated by task and consolidated according to some previously defined ER algorithm such as transitive closure or sorted neighborhood as shown by CrowdER \cite{Wang:2012} or Whang et al \cite{DBLP:journals/pvldb/WhangLG13}.
Errors made by the crowd workers are thus masked in the hope that a sufficient number of repetitions will result in the correct answer per task.
In fact, these strategies optimize for crowdsourcing cost under the assumption of a perfect crowd.
Thus, consensus is always reached and there are no contradictions in the record pairs that need to be resolved.

\smallsection{Fault-Tolerant Entity Resolution}
We term fault-tolerant entity resolution those ER mechanisms that take as input all information made available by crowd workers and build an ER solution on top of that.
In contrast to consensus-based ER, these ER mechanisms do not reject any of the crowd answers which introduces noise into the votes graph.
The challenge is then to find a clustering on top of these possibly contradicting bits of information that maximizes the agreement between crowd answers.
The work that we present in this paper builds upon preliminary work in \cite{gruenheidER2012tech}.
Related work in this category \cite{ilprints1097} has introduced a theoretic crowdsourced ER solution that uses maximum likelihood methodology to find the optimal ER solution.
This methodology is equivalent to using correlation clustering and is shown to be NP-complete.
For clustering, the authors therefore fall back onto spectral clustering and transitive closure as alternative ER mechanisms.
Their work furthermore addresses the next-crowdsourcing problem by finding those tasks that have the highest projected impact on the entity resolution solution.
This estimate extends prior work because it is not only based on the candidate positive crowd answers but also possibly negative responses.
In that respect, it is similar to our work on the next-crowdsource problem (\Autoref{sec:next}) although we examine not only uncertainty reduction strategies but explore alternative means of optimizing for error reduction in the clustering solution.

\smallsection{Fault-Tolerant Exhaustive Entity Resolution}
Exhaustive (non-monotonic) entity resolution differs from monotonic exploration of the ER space because record pairs that have been examined before can be re-evaluated at a later point in time.
Thus, it is possible to reverse a decision once made during the clustering process.
Prior work has not yet considered non-monotonic task execution for two reasons.
First, it is more expensive especially if budget is invested on `hard' tasks, i.e.,~tasks that crowd workers often disagree on.
Second, existing work commonly reasons on a task level, i.e.,~the next question to ask is not one question for task but a set of questions for one task.
As a result, once these questions have been asked, there exists sufficient signal from the crowd to determine an ER solution.
To the best of our knowledge, our work is the first to explore non-monotonic task ordering in an extensive evaluation.
In addition to sequential task execution, we also report on the trade-off between single-question execution and batch processing in \Autoref{subsec:parallel} which impacts the output quality of non-monotonic as well as monotonic crowdsourced ER.\\

\section{Data Interpretation Problem}\label{sec:faulttolerant}
A {\em decision function} is a function that determines the relationship of $r_i$ and $r_j$ based on $G$.
To correctly interpret the potentially noisy votes in $G$, it needs to fulfill several properties which we define next.
We then introduce a novel decision function called {\sc MinMax} that provides fault-tolerant data interpretation and adheres to these properties.
To contrast {\sc MinMax} with alternative pair-wise decision functions, we discuss its (dis-)advantages in \Autoref{subsec:faulttolerant:discussion}.

\subsection{Properties of Pair-Wise Decision Functions}\label{sec:pairwise}
Given a votes graph $G=(R,E)$, a desirable pair-wise decision function forms a decision about the relationship of two records $r_i, r_j \in R$ by evaluating the information contained in the edge set $E$ of the votes graph.
\begin{definition}[Decision Function]
A decision function $f$ evaluates the relationship of two records $r_i$ and $r_j$ $\in$ $R$ by first finding all distinct acyclic paths $H = (r_i, \ldots, r_j)$ connecting $r_i$ and $r_j$.
It forms its decision based on the $p_{kl}$ and $n_{kl}$ votes that are part of these paths, for each [$r_k$,$r_l$] $\in$ $H$.
The result of $f(r_i,r_j)$ is then either of three decision `yes', `no', or `unknown', which describes whether $r_i$ and $r_j$ belong to the same entity.
\end{definition}
Any decision function should obey all mathematical properties of an equivalence relation such as the `same-entity-as' relation. 
More formally, we expect the following properties from such a decision function $f$:

\smallsection{\bf Reflexivity} For any record $r_i$ and any votes graph:

\noindent$f(r_i,r_i) =$ `yes'
 
\smallsection{\bf Symmetry} For records $r_i$ and $r_j$ and any votes graph:

\noindent$f(r_i,r_j) = f(r_j,r_i)$ 

\smallsection{\bf Consistency} If the decision function decides that two records $r_i$ and $r_j$ point to the same entity, then there exists positive information between $r_i$ and $r_j$. 
Likewise, there has to exist negative information for a `no' decision.

\smallsection{\bf Convergence} For every connected record pair $[r_i,r_j]$, there has to exist an acyclic path $H = (r_i, \ldots, r_j)$ connecting $r_i$ and $r_j$ which is computable.
If $r_i$ and $r_j$ are unconnected, the default decision of the function is `unknown'.

\smallsection{\bf Transitivity} For the three records $r_i$, $r_k$, and $r_j$ and any votes graph:

\noindent$f(r_i,r_k) =$ `yes' $\wedge$ $f(r_k,r_j) =$ `yes'$\implies f(r_i, r_j) =$ `yes'
 
\smallsection{\bf Anti-transitivity} For the three records $r_i$, $r_k$, and $r_j$ and any votes graph: 

\noindent$f(r_i,r_k) =$ `yes' $\wedge$ $f(r_k,r_j) =$ `no' $\implies$ $f(r_i, r_j) =$ `no'

\noindent$f(r_i,r_k) =$ `no' $\wedge$ $f(r_k,r_j) =$ `yes' $\implies$ $f(r_i, r_j) =$ `no'\\
 
Consistency guarantees that the decision function forms appropriate decisions.
Convergence of a decision is furthermore required because it guarantees that the decision function will always make a decision.
Specifically, even if there exist cycles in the votes graph, the acyclic path generation will always generate a path between records $r_i$ and $r_j$.
(Anti-)transitivity is an essential tool for cost-conscious environments such as ER in a crowdsourcing setup.
Decision functions apply it because in contrast to traditional ER, each information request incurs additional monetary cost which is to be avoided.
As a result, both transitivity and anti-transitivity are core concepts of decision functions in state-of-the-art ER solutions \cite{DBLP:conf/sigmod/WangLKFF13}.

\subsection{MinMax Similarity Measure}
The {\sc MinMax} pair-wise decision function $f_{M}$ uses both positive and negative information in the crowd workers' input and establishes a similarity measurement for every record pair through path-based inference.
It is a novel fault-tolerant technique that is inspired by work on preference functions \cite{Fagin} and voting schemes that discuss ranking pair-wise decisions \cite{schulze11}.
While these mechanisms are used for decision-making in a space where crowd signals are one-dimensional (i.e.,~$r_i$ and $r_j$ belong together or $r_i$ is better than $r_j$), obtaining information from the crowd for ER problems enables both positive and negative decision signals (i.e.,~$r_i$ and $r_j$ belong or do not belong together). 
{\sc MinMax} uses the notion of {\it positive} and {\it negative acyclic paths} in the votes graph to evaluate whether two records belong to the same entity.
The score of a positive path between records $r_i$ and $r_j$ is denoted as $p^*_{ij}$ while the negative path scores are referred to as $n^*_{ij}$.
\begin{definition}[MinMax Decision Function]
Given two records $r_i$ and $r_j$ decides whether $r_i$ and $r_j$ belong to the same entity given the positive votes $p^*_{ij}$ and the negative votes $n^*_{ij}$ along the path(s) connecting these two records.
$$
f_M (r_i,r_j) =
\begin{cases}
\mbox{Yes,} \hspace{40pt} \> \mbox{if } p^*_{ij} - n^*_{ij} \geq q_p \\
\mbox{No,} \hspace{42pt} \> \mbox{if } n^*_{ij} - p^*_{ij} \geq q_n \\
\mbox{Do-not-know,} \hspace{10pt} \> otherwise
\end{cases}
$$
\end{definition}
{\sc MinMax} decides that $r_i$ and $r_j$ belong to the same entity if there is sufficient evidence for such a decision represented through a quorum $q_p$ ($q_n$ for negative decisions).
The higher $q_p$ and $q_n$, the more crowd workers need to support a positive (negative) decision which ensures better result quality.
In practice, we show that a quorum as low as 3 is sufficient for accurate decision making as shown in \Autoref{sec:experiments}.

\begin{definition}[Path Definition]
A positive path in a votes graph is a sequence of records $H_p = (r_i, \ldots, r_j)$ connecting records $r_i$ and $r_j$ such that all consecutive record pairs [$r_k$,$r_l$]$\in H$ have only positive weights, i.e.,~$p_{kl}>0$.
A negative path $H_n = (r_i, \ldots, r_j)$ contains exactly one negative record pair [$r_k$,$r_l$] such that $n_{kl}>0$.
All other record pairs on $H \backslash [r_k,r_l]$ have positive weights with $p_{kl}>0$.
\end{definition}
The notion of paths allows for effective transitivity:
If two records $r_i$ and $r_j$ are connected through a positive path with high weights and no negative paths, it is likely that they belong to the same entity.
Anti-transitivity on the other hand can be leveraged with exactly one negative edge in the path only:
If $r_i \neq r_k$ and $r_k \neq r_j$, there is no way to automatically infer the relationship of [$r_i$,$r_j$].
\begin{example}[Paths]\label{ex:paths}
Going back to the introductory example (\Autoref{ex:animals}), imagine that we want to determine the positive and negative paths connecting $r_2$ and $r_3$.
There exist negative paths $H_{n_1} = (r_2,r_4,r_3)$ and $H_{n_2} = (r_2,r_1,r_3)$.
There also exists exactly on positive path $H_p = (r_2,r_1,r_3)$ that has only positive path scores.
\end{example}
The core idea of the {\sc MinMax} decision mechanism is that a path is only as strong as its weakest link, i.e.,~the edge in the path that has the lowest weight.
Both positive and negative paths are thus assigned a score according to the {\it minimum} weight of all absolute edge weights.
If there exist multiple paths that connect two records, we choose the path with the {\it maximal} score as it signifies a higher confidence from the crowd.
We call the path scores for a record pair $[r_i,r_j]$ the {\it positive} score $p^*_{ij}$ resp. the {\it negative} score $n^*_{ij}$ of a path.
\begin{definition}[Path Scores]\label{def:PathScore}
Given records $r_i$ and $r_j$ $\in$ $R$, let $H^p_{ij}$ denote the set of positive paths connecting $r_i$ and $r_j$.
For every path $H_{\mu} \in H^p_{ij}$, its score $h(H_{\mu})$ is computed as the minimum of all direct scores for any consecutive record pair [$r_k$,$r_l$], i.e.,~$h(H_{\mu})$=min($p_{kl}$), for all [$r_k$,$r_l$] $\in$ $H$.
The {\sc MinMax} path score $p^*_{ij}$ is then computed as max($h(H_{\mu})$ $\forall$ $H_{\mu} \in H^p_{ij}$).
Negative path scores $n^*_{ij}$ are computed analogously.
\end{definition}
Positive and negative {\sc MinMax} scores can be represented as a two-dimensional matrix of record pairs:
Positive path scores correspond to the upper triangular part of the matrix while negative paths to the lower triangular.
If there exists no positive or negative path between two records, $p^*_{ij}$ resp. $n^*_{ij}$ is set to 0.
\begin{figure}[t]
\begin{subfigure}{\columnwidth}
\caption{Internal representation (Example \ref{ex:MMexample}).}\label{fig:MMexample}
\end{subfigure}
\centering

\begin{minipage}{.4\columnwidth}
\centering
\includegraphics[width=.62\columnwidth]{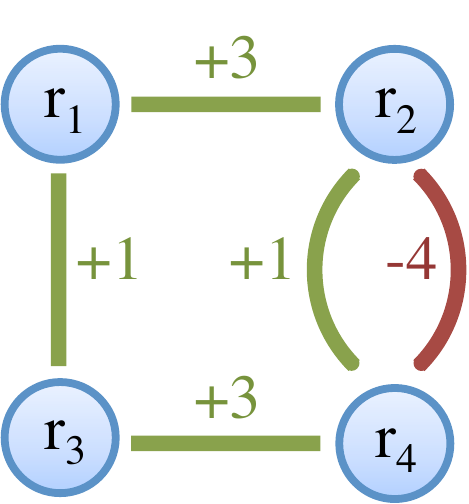}
\end{minipage}
\begin{minipage}[t]{.48\columnwidth}
\small
\begin{tabular}{c|c c c c}
\multicolumn{1}{c|}{}& {$r_1$} & {$r_2$} & {$r_3$} & {$r_4$}\\
\hline
$r_1$ & - & \color{ForestGreen}{+3} & \color{ForestGreen}{+1} & \color{ForestGreen}{+1}\\
$r_2$ & \color{BrickRed}{-1} & - & \color{ForestGreen}{+1}1 & \color{ForestGreen}{+1}\\
$r_3$ & \color{BrickRed}{-3} & \color{BrickRed}{-3} & - & \color{ForestGreen}{+3}\\
$r_4$ & \color{BrickRed}{-3} & \color{BrickRed}{-4} & \color{BrickRed}{-1} & -\\
\end{tabular}
\end{minipage}

\vspace{.05in}
\small{Votes Graph \hspace{.75in} {\sc MinMax} Matrix}
\vspace{.05in}

\centering
\begin{minipage}{.4\columnwidth}
	\centering
	\includegraphics[width=.62\columnwidth]{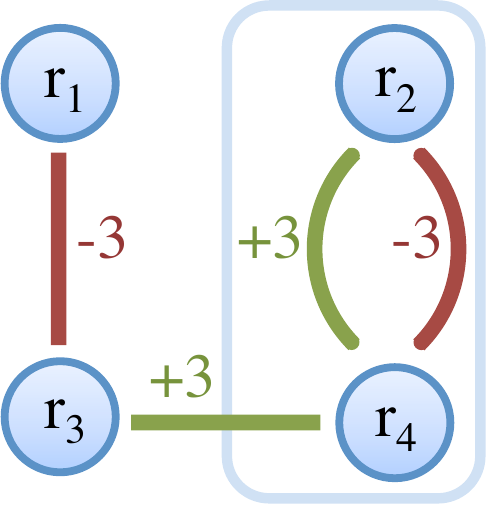}
\end{minipage}
\begin{minipage}[t]{.48\columnwidth}
\small
\begin{tabular}{c|c c c c}
\multicolumn{1}{c|}{}& {$r_1$} & {$r_2$} & {$r_3$} & {$r_4$}\\
\hline
$r_1$ & - & {\bf 0} & 0 & 0\\
$r_2$ & \color{BrickRed}{-\bf 3} & - & \color{ForestGreen}{+3} & \color{ForestGreen}{+3}\\
$r_3$ & \color{BrickRed}{-3} & \color{BrickRed}{-3} & - & \color{ForestGreen}{+3}\\
$r_4$ & \color{BrickRed}{-3} & \color{BrickRed}{-3} & 0 & -\\
\end{tabular}
\end{minipage}
\begin{subfigure}{\columnwidth}
\vspace{.03in}
\caption{Noise detection (Example \ref{ex:MMproblem}).}\label{fig:MMproblem}
\end{subfigure}

\caption{{\sc MinMax} computation.}
\end{figure}
\begin{example}[MinMax Computation]\label{ex:MMexample}
To see how the {\sc MinMax} scores are computed, take a look at \Autoref{fig:MMexample} where green edges signify positive and red edges equal negative votes.
To make a decision for $[r_3,r_4]$, two positive paths are examined:
$(r_3,r_4)$ with a score of 3 and $(r_3,r_1,r_2,r_4)$ with min(1,3,1)=1. 
The score for $p^*_{(3,4)}$ is then computed as max(3,1)=3.
To compute the negative score $n^*_{(3,4)}$, we observe that there exists only one negative path, $(r_3,r_1,r_2,r_4)$. 
The negative path score $n^*_{(3,4)}$ is thus min(1,3,4)=1.
\end{example}
Computing both, negative and positive, paths ensures fault-tolerant data interpretation.
At the same time, it also gives the {\sc MinMax} decision mechanism the power to adapt to inconsistencies found in crowd responses:
Low confidence decisions can be outweighed through stronger evidence for the opposite decision.
\begin{example}[Fault-Tolerance]\label{ex:MMfault}
Imagine that in \Autoref{fig:MMexample}, the edges between $r_2$ and $r_4$ do not yet exist.
The positive path connecting all records then indicates that they all belong to the same entity albeit with low confidence, i.e.,~$p^*_{24}$=min(3,1,3)=1.
Introducing the new edge evidence between $r_2$ and $r_4$ weakens the previously made decision and causes the entity to be split into two separate entities.
The negative edge between $r_2$ and $r_4$ now propagates to $r_1$ ($n^*_{14}$=min(4,3)=3) and $r_3$ ($n^*_{23}$=min(4,3)=3) respectively.
\end{example}
Choosing the weakest link of a path provides a mechanism to cautiously evaluate a path.
Nevertheless, using absolute values has one general drawback:
If a decision is based on a positive score outweighing a negative score or vice versa by a specified margin, {\it dominating} values can cause incorrect inference as a result of incomplete information.
Thus, an algorithm relying on such information may make a decision for [$r_i$,$r_j$] for which intutively a decision should not be made.
Imagine two records that are connected through a positive edge ($r_k$,$r_l$) which dominates the path score, i.e.,~it is the minimal edge in the path.
We say that ($r_k$,$r_l$) is {\it dominated} if its counterpart, here the corresponding negative edge between $r_k$ and $r_l$, has a higher weight.
Intuitively, this kind of inference should not be allowed as it propagates wrong beliefs:
The votes indicate that a negative decision should be made but instead, the positive information is used for the path computation and therefore propagated into the path scores.
This may lead to false conclusions as shown in the following example.
\begin{example}[Dominating Votes]\label{ex:MMproblem}
In \Autoref{fig:MMproblem}, $r_1$ is unequal to $r_3$ which is positively connected to $r_4$.
Additionally, the crowd is unsure whether $r_4$ is unequal to $r_2$.
Therefore, it should not be possible to infer whether $r_1$ and $r_2$ belong to the same or different entities.
Computing the path scores as previously shown, {\sc MinMax} proposes a negative score $n^*_{12}$=3 for the record pair ($r_1,r_2$) but as there exists no positive path between $r_1$ and $r_2$, $p^*_{12}$=0.
As a result, the relationship of ($r_1$,$r_2$) is falsely estimated.
\end{example}

\subsection{Properties of MinMax}\label{sec:propminmax}
To counter the effect described in \Autoref{ex:MMproblem}, we modify {\sc MinMax} to limit the propagation of dominated edges.
More specifically, we can only compute a path between two records $r_i$ and $r_j$ if any edge on the path is not dominated by its opposite edge.
\begin{corollary}[Edge Domination]\label{cor:edgedom}
A vote $p_{ij}$ ($n_{ij}$) for a record pair $[r_i,r_j]$ can be used for path computation only if it dominates $n_{ij}$ ($p_{ij}$), i.e.,~$p_{ij}>n_{ij}$ ($p_{ij}<n_{ij}$).
\end{corollary}
Next to avoiding wrong propagation of values, this corollory also guarantess that the path computation of {\sc MinMax} adheres to a weaker form of (anti-)transitivity.
Remember that when working with erroneous data, conflicting information may lead to interpretation inconsistencies, i.e.,~there exist paths between records $r_i$ and $r_j$ that convey contradicting decisions, one positive and one negative.
Weak (anti-)transitivity guarantees that the decision function proceeds {\it cautiously} when encountering these noisy path values.
Instead of risking a faulty decision, the decision function will return `unknown' instead.

\begin{definition}[Weak Transitivity]
For records $r_i$, $r_k$, and $r_j$ and any votes graph:

\noindent$f(r_i,r_k) =$ `yes' $\wedge$ $f(r_k,r_j) =$ `yes'

\hspace{.24in} $\implies$ $f(r_i, r_j) =$ `yes' $\lor$ $f(r_i, r_j) =$ `unknown'
\end{definition}
\noindent Weak anti-transitivity is defined analogously.

\begin{proposition}[Weak Transitivity of {\sc MinMax}]

\noindent{\sc MinMax} is weakly transitive.
\end{proposition}

\begin{proof}[Weak Transitivity of {\sc MinMax}]\label{proof:weaktransitivity}
\abovedisplayskip=1pt
\belowdisplayskip=1pt
Given three records $r_i$, $r_j$, and $r_k$ $\in$ $R$ transitivity is violated if $f(r_i, r_k)$ = `yes', $f(r_k, r_j)$ =
`yes' but $f(r_i, r_j)$ = `no'.
Path scores $p^*_{ij}$ and $n^*_{ij}$ are computed as follows:
\begin{align}
p^*_{ij} &= \min \{p^*_{ik}, p^*_{kj}\} \label{eq:agg_positive}\\
n^*_{ij} &= \max \Big\{ \min \{n^*_{ik}, H^p_{kj}\}, \min \{H^p_{ik},n^*_{kj}\} \Big \} \label{eq:agg_negative}
\end{align}
Equation~\ref{eq:agg_positive} follows from the definition of positive path scores. 
Note that if there is any direct edge between $r_i$ and $r_j$, it can only improve the positive score $p^*_{ij}$ due to the maximum computation of {\sc MinMax}, which makes Equation~\ref{eq:agg_positive} the worst case scenario for positive path computation.
Computing $n^*_{ij}$ means evaluating each negative subpath going through $r_k$ with all possible positive paths $H^p$ that complete the path from $r_i$ to $r_j$.
Any direct edge between $r_i$ and $r_j$ is already considered in this computation because it is reflected in the subpath scores. For example, the
computation of $n^*_{ik}$  already considers all paths between $r_i$ and $r_k$ including those that that contain a direct edge 
between $r_i$ and $r_j$.

If transitivity is violated then $p^*_{ij} < n^*_{ij}$ must hold. However, this contradicts {\sc MinMax} as described in the following case analysis designed according to Equation~\ref{eq:agg_positive} and \ref{eq:agg_negative}:

\noindent{{\bf Case 1}: $p^*_{ij} = p^*_{ik}$}
\begin{itemize}
\item Case 1.1: $n^*_{ij} = \min \{n^*_{ik}, H^p_{kj}\}$. Under the transitivity violation assumption, this means that
$p^*_{ik} \leq n^*_{ik}$ which violates \Autoref{cor:edgedom} and the initial assumption that $f(r_i, r_k)$ = `yes'.
	\item Case 1.2: $n^*_{ij} = \min \{H^p_{ik}, n^*_{kj}\}$. This means that $p^*_{ik} \leq  H^p_{ik}$ for any positive path $H^p_{kj}$.
	A strict inequality contradicts {\sc MinMax} and the positive path score definition as $p^*_{ik}$ is the strongest path between $r_i$ and $r_j$.
	An equality relationship is achieved for example if there is only one path or all paths have the same weight.
	As a result, it is possible to obtain $p^*_{ij} = n^*_{ij}$, thus $f(r_i, r_j)$ = `unknown'.
	This behavior is the reason, why {\sc MinMax} only guarantees {\it weak} transitivity.
\end{itemize} 
Therefore, weak transitivity cannot be violated if path scores are computed as in Definition~\ref{def:PathScore} and if {\sc MinMax} is used as a decision function.

\noindent{{\bf Case 2}: $p^*_{ij} = p^*_{kj}$} The proof for this condition is analogous to Case 1.
\end{proof}

\smallsection{\bf Weak Anti-transitivity} For records $r_i$, $r_k$, and $r_j$ and any votes graph:

\noindent$f(r_i,r_k) =$ `yes' $\wedge$ $f(r_k,r_j) =$ `no'

\hspace{.24in} $\implies$ $f(r_i, r_j) =$ `no' $\lor$ $f(r_i, r_j) =$ `unknown'\\

\noindent Weak anti-transitivity can be proven in a similar fashion. Violating weak anti-transitivity
(\emph{i.e.} $f(r_i, r_j) =$ `yes' and therefore $p^*_{ij} > n^*_{ij}$) contradicts the {\sc MinMax} definition as well
as the assumption that $f(r_k,r_j) =$ `no'.
Even though {\sc MinMax} adheres to a weaker variant of transitivity, we show experimentally that transitivity is in practice leveraged efficiently in \Autoref{sec:experiments}.

Furthermore, compared to a simple majority-based decision function, i.e.,~a function that decides in favor of the majority for every record pair $[r_i,r_j]$ ({\sc MA}), {\sc MinMax} obtains higher result quality due to this property.
Quality in this context is defined as the pair-wise accuracy of records within the entity resolution solution.
For example, if $r_i$ and $r_j$ are in the same cluster according to the ground truth but are put into different clusters by the ER algorithm then the accuracy of the result ER solution decreases.
To understand why we claim that {\sc MinMax} will eventually lead to the same or a better quality than a majority-based approach that builds upon dominating edges, let us revisit the running example in \Autoref{fig:animals}.
Here, we have observed one false decision between the baby animals $r_1$ and $r_3$.
A majority-based algorithm would at this point form the decision `no' and move on to the next record pair.
However, the {\sc MinMax} path scores in this example are $p^*_{(1,3)}=3$ and $n^*_{(1,3)}=2$ and if quorum $q_n$ is bigger than 1, then this record pair is marked as unresolved.
As a result, the system would require more information from the crowd workers to form a decision.
In the worst case, this decision would not change making the final decision equivalent to that of the majority approach.
In the best case, new workers would distinguish the species and increase the number of positive votes.

\begin{corollary}[MinMax Result Quality]
The result quality of {\sc MinMax} is at least as good as the result quality of a majority-based decision function.
\end{corollary}
More formally, under the assumption of non-malicious task workers, i.e.,~worker answers correctly with probability $p$~$>$~0.5, the following observation holds:
Given a vote set $E_i$ at timestamp $t_i$ and an enriched vote set $E_{i+1}$ at $t_{i+1}$ with $|E_i|$ $<$ $|E_{i+1}|$, the decisions made by {\sc MinMax} at $t_{i+1}$ are more accurate than at $t_i$.
Imagine that at $t_i$, a majority-based approach {\sc MA} and {\sc MinMax} compute their ER solution.
Either decision made by {\sc MinMax} is then the same as made by {\sc MA} or `unknown' due to the weak transitivity of {\sc MinMax}.
To resolve unknown decisions, {\sc MinMax} enriches the edge set to $E_{i+1}$.
The ER solution of {\sc MinMax} based on $E_{i+1}$ has at least the same quality as for $E_{i}$ because $p$ $>$ 0.5.
As a result, the decisions made by {\sc MinMax} at $t_{i+1}$ are at least as good as the decisions made by {\sc MA} at $t_i$, i.e.,~{\sc MinMax} might incur higher cost but will not have worse quality than {\sc MA}.


\subsection{Discussion}\label{subsec:faulttolerant:discussion}
{\sc MinMax} is a decision function computed on absolute values, associating record pairs with at most two numerical values that describe the relationship of the records in this pair.
In contrast to working with absolute values, relative decision functions, i.e.,~computing the score of a record pair based on their distance and connectivity, can be used.
The drawback of these functions is that they require to compute and maintain {\it all} paths between record pairs which is inefficient to execute in practice, especially in highly connected graphs.
Furthermore, they have to be selected carefully as they need to fulfill the graph properties described in \Autoref{sec:pairwise}.
This cannot be guaranteed for functions such as {\it sum} and {\it count} for example, \cite{gruenheidER2012tech}.
Probabilistic decision functions are another set of functions that have been explored in \cite{DBLP:journals/pvldb/WhangLG13} amongst others where the authors also show that they are infeasible to compute even with a perfect crowd.
For the remainder of this paper, we will thus focus on {\sc MinMax} as a representative of a group of decision functions that adhere to the presented graph properties and provide minimal computational overhead.

\section{Clustering Algorithm}\label{sec:er}
Traditional ER algorithms use the pair-wise information in a graph to construct a clustering that represents the final ER solution.
This methodology is defined in the framework aglrotihm (Algorithm~\ref{alg:generalER}) as the data interpretation problem.
In the last section, we have discussed how we can find pair-wise information.
Now, we examine how we can efficiently cluster records into entities.
Example algorithms herefore are cut or correlation clustering, \cite{DBLP:journals/ml/BansalBC04}, which we adapt in the following to suit the incomplete resolution space that is inherent to incremental ER.
Before going into detail on the applied ER algorithm, we first establish how the votes graph and the {\sc MinMax} matrix can be adapted when new votes are collected from the crowd.
The observations that we make are specific to {\sc MinMax} but can be easily adapted to any decision function that has the properties defined in \Autoref{sec:faulttolerant}.

\begin{figure}
\begin{subfigure}{.3\columnwidth}
\includegraphics[width=\columnwidth]{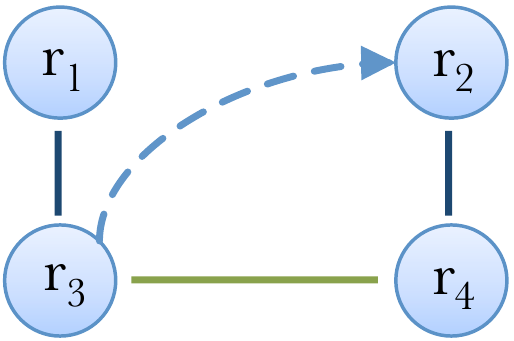}
\caption{$r_3 \rightarrow r_2$}
\end{subfigure}
\hfill
\begin{subfigure}{.3\columnwidth}
\includegraphics[width=\columnwidth]{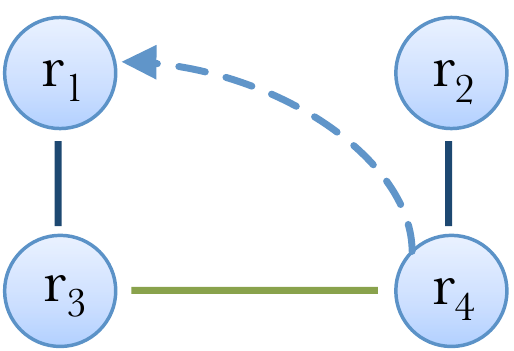}
\caption{$r_4 \rightarrow r_1$}
\end{subfigure}
\hfill
\begin{subfigure}{.3\columnwidth}
\includegraphics[width=\columnwidth]{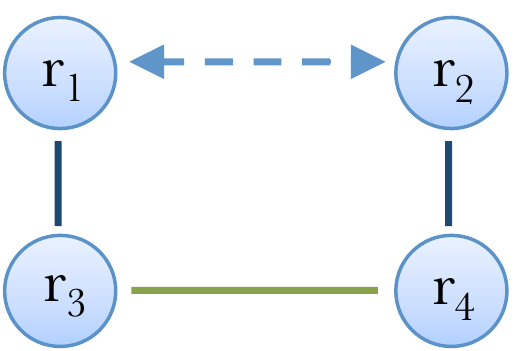}
\caption{$r_1 \leftrightarrow r_2$}
\end{subfigure}
\caption{Incremental update propagation.}\label{fig:updatepropagation}
\end{figure}

\subsection{Updating the Votes Graph}\label{subsec:update}
As explained in \Autoref{sec:problem}, incremental ER requires an adjustment of its queue and its ER solution whenever new information becomes available.
To understand the ER update process, i.e.,~the integration of a new positive or negative vote, first imagine how an update is propagated:
Obviously, an update affects the records that are directly modified.
From there on, it may affect connected records, which then further propagate the changes through the votes graph.
This notion is visualized in Figure \ref{fig:updatepropagation} for the running example where an edge between records $r_3$ and $r_4$ is modified while the edges between $r_1$ and $r_3$ resp. $r_2$ and $r_4$ remain untouched.
As a result, the following three update operations need to be explored:
\begin{enumerate}
  \item All paths connecting $r_3$ to any other record through $r_4$ may be modified.
  \item All paths connecting $r_4$ to any other record through $r_3$ may be modified.
  \item All paths that are connected through the new edge between $r_3$ and $r_4$ may be modified.
\end{enumerate}
The records that are (indirectly) modified by the changes are called the {\it update component} of the update.
They are those records for which either the positive or the negative path connecting them to any other record has been modified directly or indirectly through the update of an edge.
\begin{theorem}[Update Component]
The changes due to a new or modified edge $e_i \in E$ need to be propagated only along those edges $e_j \in E$ for which either $p^*_{ij}$ or $n^*_{ij}$ changes. 
\end{theorem}
\begin{proof}
The validity of this theorem derives from the fact that an update is a local change that affects the global computation of paths by modifying their subpaths.
Therefore, if a path is not changed, then the scores of other paths that use this subpath will not change.
The propagation of updates along edges is analogous to the propagation of change described in \cite{DBLP:journals/pvldb/GruenheidDS14}.
\end{proof}
The update functionality used to propagate changes for {\sc MinMax} is described in Algorithm~\ref{alg:recUpdate}.
Here, the increment computation is split into two parts:
Algorithm {\sc Update} describes how an increment can be realized, analogous to what is visualized in Figure~\ref{fig:updatepropagation}; Function {\sc Compute} then propagates the update along all affected paths recursively.
The actual update computation uses a boolean path marker $\gamma$ for declaring negative and positive paths (Lines \ref{alg:line:computestart} - \ref{alg:line:computeend}) as follows:
\begin{itemize}[leftmargin=.05\columnwidth]
  \item {\bf If $\gamma = false$}, the algorithm knows that the current path is negative. It therefore evaluates only positive outgoing edges for nodes that have not been traversed on this path (Line \ref{alg:line:validneg}).
For each of these candidate paths, {\sc Compute} is called again to propagate the update further.
\item {\bf If $\gamma = true$}, the algorithm follows the same procedure as above but has to consider all outgoing edges, independent of whether they are positive or negative.
As a result, it is recursively called for both scenarios (Lines \ref{alg:line:candidateallpos}-\ref{alg:line:candidateallneg}).
\end{itemize}
The {\sc Compute} function is used by Algorithm {\sc Update} for all three update scenarios described previously.
The first two scenarios (Lines \ref{alg:line:compute1} and \ref{alg:line:compute2}) explore all modified paths in the respective update direction.
Solving the third scenario (Line \ref{alg:line:compute3}) is more complex as all path alternatives that provide the {\sc MinMax} score for a record pair ($r_i,r_j$) need to be stored.
As this is infeasible to compute, only the path between $r_i$ and $r_j$ that has the highest path score, referred to as $path_{(r_i,r_j)}$, is stored.
In case there exist multiple paths with equal scores, the shortest path is chosen.
Obviously, this approximation allows for impreciseness in the result but as we show in our experimental evaluation, this simplification has no discernible impact on the result quality.
After the changes to the {\sc MinMax} matrix have been computed, we adjust the current ER solution by running a clustering algorithm on the set of modified records and their corresponding clusters (Line \ref{alg:line:cluster}).
\begin{algorithm2e}[t]
\caption{MinMax incremental update algorithm {\sc update} and supporting function {\sc compute}.}\label{alg:recUpdate}
\small
{\bf Algorithm: }{\sc update($G = (R,E)$, ($r_i$,$r_j$))} : void

compute\_new\_paths($r_i$,$r_j$)\;
\nonl\hspace{.02in}\textcolor{gray}{// compute paths of form $r_i$ -- $r_j$ -- $r_k$}\\
$R_i \leftarrow$ {\sc compute}($r_j$,\{$r_i$\},$w_{ij}>0$)\;\label{alg:line:compute1}
\nonl\hspace{.02in}\textcolor{gray}{// compute paths of form $r_j$ -- $r_i$ -- $r_k$}\\
$R_j \leftarrow$ {\sc compute}($r_i$,\{$r_j$\},$w_{ij}>0$)\;\label{alg:line:compute2}
\nonl\hspace{.02in}\textcolor{gray}{// compute paths of form $r_k$ -- $r_i$ -- $r_j$ -- $r_l$}\\
\lForEach{$r_k \in R_i$} {
	$R_k \leftarrow$ {\sc compute}($r_j$,\{$path_{r_k,r_j}$\},$w_{kj}>0$)\;\label{alg:line:compute3}
}
{\sc resolve}($R_i \cup R_j \cup R_k$)\;\label{alg:line:cluster}

\nonl\hrulefill

{\bf Function: }{\sc compute}$(r_i,R_i,\gamma)$ : void\label{alg:line:computestart}

\uIf{$\gamma$ = false} {
	\nonl\hspace{.02in}\textcolor{gray}{// compute negative paths}\\
	$R_j \leftarrow$ G.get\_records($r_i$,$R_i$) : $\forall r_j \in R_j: r_j \notin R_i \wedge p_{ij}>0$\;\label{alg:line:validneg}
	\ForEach{$r_j \in R_j; $} {
		$R_n \leftarrow$ compute\_negative\_paths($r_i$,$r_j$)\;\label{alg:line:candidateneg}
		\lForEach{$r_n \in R_n$} {{\sc compute}($r_n$,$R_i \cup r_i$,false)}
	}
}
\Else {
	\nonl\hspace{.02in}\textcolor{gray}{// compute negative and positive paths}\\
	$R_j \leftarrow$ G.get\_records($r_i$,$R_j$) : $\forall r_j \in R_j: r_j \notin R_i$\;\label{alg:line:validall}
	\ForEach{$r_j \in R_j$} {
		$R_n \leftarrow$ compute\_new\_paths($r_i$,$r_j$)\;\label{alg:line:candidateall}
		\ForEach{$r_n \in R_n$} {
			\lIf{$p_{ij} > 0$}{\sc compute}($r_n$,$R_i \cup r_i$,true)\;\label{alg:line:candidateallpos}
			\lIf{$n_{ij} > 0$}{\sc compute}($r_n$,$R_i \cup r_i$,false))\label{alg:line:candidateallneg}
		}
	}
}\label{alg:line:computeend}

{\bf return} $R_n$\;
\end{algorithm2e}
\begin{example}[{\sc update} Algorithm]\label{ex:inccmm}
Going back to the initial example (Figure~\ref{fig:MMexample}), imagine that the edge between $r_3$ and $r_4$ is newly inserted.
The algorithm then first explores the positive path between $r_3$ and $r_2$ but as $p^*_{(2,3)}$ is already 1, it is not further pursued.
In contrast, a negative path connecting the records did not previously exist, thus $n^*_{(2,3)}$ is set to 3.
Following the connecting positive edge to $r_1$, the algorithm then sets $n^*_{(1,3)}$ to 3.
As no path value connecting $r_4$ to any other record via $r_3$ is changed, the update process finishes.
\end{example}
\begin{theorem}[Complexity of {\sc update}]
Updating the {\sc MinMax} matrix has a complexity of O($m^2$) in the worst case where $m$ represents the number of records in the update component.
\end{theorem}
\begin{proof}
Updates connecting records $r_i$ and $r_j$ are propagated in three different ways as described previously.
The number of paths in the update component is maximized if the update affects $\frac{m}{2}$ in the first and second propagation phase (Lines \ref{alg:line:compute1} and \ref{alg:line:compute2}).
In the third phase, at most $\frac{m}{2} ^* \frac{m}{2}$ pair-wise scores need to be updated which results in a quadratic worst-case execution time. 
\end{proof}
Given the cost-optimized incremental structure of the votes graph, the average execution time is significantly lower in practice.

\subsection{Clustering Records}\label{subsec:er}
Using up-to-date positive and negative path scores, an ER solution based on the current state of the votes graph and {\sc MinMax} matrix can be computed.
For that purpose, we employ a variant of correlation clustering which adheres to the concepts of intra-cluster densitiy and inter-cluster sparsity.
To evaluate whether two records $r_i$ and $r_j$ refer to the same entity, the benefit and penalty of having them in the same cluster are weighed off according to their $p^*_{ij}$ and $n^*_{ij}$ scores:
If $p^*_{ij}>n^*_{ij}$ then it is more beneficiary to assign both to the same entity, otherwise they are assigned to different entities.
We use a variation of an established approximation for clustering records called {\it cautious} correlation clustering \cite{DBLP:journals/ml/BansalBC04} (Algorithm \ref{alg:cluster}).
The {\sc Resolve} algorithm proceeds as follows:
It first randomly selects a record $r_i$ that has not been assigned to a cluster.
It is then added to a new cluster $c_i$ together with all those records $r_j$ that it is positively connected to, i.e.,~$p^*_{ii}>n^*_{ij}$.
Record $r_j$ can only become a part of $c_i$ if it is not a part of a cluster yet.
It next initiates a vertex removal phase in which all those records are removed from $c_i$ where the penalties of keeping them in $c_i$ outweigh the benefits (Line \ref{alg:cluster:remove}).
For example if a $r_j$ has a positive relationship with weight 1 to $r_i$ but a negative relationship with weight 2 with $r_k \in c_i$, it is better to remove $r_j$ from $r_k$ if the score of fulfills $[r_i,r_k]>1$.
The `goodness' of record $r_j$ is examined in function {\sc isGood} which is a weighted validity computation analogous to the computation of $\delta$-goodness in \cite{DBLP:journals/ml/BansalBC04}.
After the vertex removal phase has been completed, a record addition phase is initiated.
Here, any records are added to $c_i$ if the benefit of adding that record outweighs the penalty (Line \ref{alg:cluster:add}).

\begin{algorithm2e}[t]
\caption{Weighted path ER algorithm {\sc resolve} and supporting function {\sc isGood}.}\label{alg:cluster}
\small
{\bf Algorithm: } {\sc resolve}$(G = (R,E))$ : ER solution $C$

{\ForEach{$r_i \in R; r_i \notin C$\label{alg:cluster:iterate}} {
	$c_i.add(r_i)$\;
	$c_i.addAll(r_j)$ : $r_j \in Q$; $p^*_{ij}>n^*_{ij}$\;
	\nonl\hspace{.02in}\textcolor{gray}{// record removal phase}\\
	\ForEach{$r_k \in c_i$} {\label{alg:cluster:removeloop}
		\If{$\neg$ {\sc isGood}($r_k,c_i$)}{
			$c_i$.remove($r_k$)\;\label{alg:cluster:remove}
			reset loop in Line \ref{alg:cluster:removeloop} and jump there\;
		}
	}
	\nonl\hspace{.02in}\textcolor{gray}{// record addition phase}\\
	\ForEach{$r_k \in R; r_k \notin c_i \wedge \exists r_i \in c_i: p^*_{ik}>n^*_{ik}$} {\label{alg:cluster:addloop}
		\If{{\sc isGood}($r_k,c_i$)}{
			$c_i$.add($r_k$)\;\label{alg:cluster:add}
			reset loop in Line \ref{alg:cluster:addloop} and jump there\;
		}
	}
	\If{$c_i = \emptyset$}{
		$c_i.add(r_i)$\;
	}
	$C.add(c_i)$\;\label{alg:cluster:addcluster}
}}
{\bf return} C;

\nonl\hrulefill

{\bf Function: }{\sc isGood}$(r_i,c_i)$ : boolean

{score, penalty = 0 \;}
\ForEach{$r_j \in c_i$;  $r_i\neq r_j$} {
	\lIf{$p^*_{kl}>n^*_{kl}$} {
		score += $p^*_{kl} - n^*_{kl}$\;
	}
	\lElse {
		penalty += $n^*_{kl} - p^*_{kl}$\;
	}
}
{\bf return} score$>$penalty\label{alg:cluster:good}
\end{algorithm2e}

\begin{example}[{\sc resolve} Algorithm]\label{ex:recClust}
In \Autoref{fig:MMexample}, the votes graph contains records $R = \{r_1,r_2,r_3,r_4\}$.
Picking $r_1$ randomly, $r_1$ is added in addition to $r_2$, with $p^*_{(1,2)}>n^*_{(1,2)}=3>1$, to cluster $c_1$.
Removing either $r_1$ or $r_2$ from the cluster does not provide a better clustering score, thus the cluster remains unchanged.
Neither $r_3$ nor $r_4$ are added to $c_1$ as for both the negative path values to any record in $c_1$ is at least as high as the respective positive scores.
Similarily, the algorithm constructs cluster $c_2$ containing both $r_3$ and $r_4$.
Thus, the final clustering is [\{$r_1,r_2$\},\{$r_3,r_4$\}].
\end{example} 
Instead of running {\sc resolve} on all records after an update, it is only run on those records that are in any cluster that has been touched by an update (Algorithm \ref{alg:recUpdate}, Line \ref{alg:line:cluster}).
Since these records do not necessarily represent whole clusters, the algorithm first expands the update component into a {\it transitive} update component, i.e.,~the original update component and all records that are in any of the clusters contained in the update component.
Partial ER clustering for the transitive update component is triggered after each update to the votes graph to provide a consistent ER solution.
\begin{theorem}[Complexity of {\sc resolve}]
Running entity resolution on the updated parts of the {\sc MinMax} matrix has a worst case complexity in $O(|R|m^*)$.
\end{theorem}
\begin{proof}
The vertex removal phase requires $m^*$ repetitions at most where $m^*$ is the number of records in the transitive update component.
The vertex addition phase is bounded by the number of overall records in the graph, $|R|$.
As both phases are initiated for all $m^*$ records and $m^*<|R|$, {\sc resolve} runs in $O(|R| m^*)$.
\end{proof}

\eat{
\begin{figure}[t]
\centering
\pgfplotstableread{graphdata/synthetic_time.dat}{\synthtime}
\hspace*{-.5in}
\begin{subfigure}[b]{.7\columnwidth}
\begin{tikzpicture}[font=\small]
	\small
	\begin{axis}
	[xlabel=Crowd Accesses (in K),
	ylabel= Update Time (in ms),
	legend style={at={(axis cs:130,5)},anchor=south west, nodes=right, draw=none, outer sep = 0pt, inner sep = 0pt, font=\tiny},
	xticklabel style={/pgf/number format/.cd,fixed,precision=2},
    xtick={0,40,80,120,160,200},
    ytick={0,50,100,150,200},
    width=\columnwidth,
    height=.6\columnwidth
    ]
	\addplot [green2, very thick, mark options=solid] table [x={cost100}, y={time100}] {\synthtime};
	\addlegendentry{$n$=100}
	\addplot [blue1, dotted,very thick, mark options=solid] table [x={cost250}, y={time250}] {\synthtime};
	\addlegendentry{$n$=250}
	\addplot [green1, dashed, very thick, mark options=solid] table [x={cost500}, y={time500}] {\synthtime};
	\addlegendentry{$n$=500}
	\addplot [red1, densely dotted,very thick, mark options=solid] table [x={cost1000}, y={time1000}] {\synthtime};
	\addlegendentry{$n$=1000}

	\end{axis}%
\end{tikzpicture}
\end{subfigure}
\caption{Time measurements for synthetic datasets with varying number of records $n$ and a perfect crowd.}\label{fig:synthtime}
\end{figure}}

\begin{table}
\caption{Time measurements for synthetic datasets with varying number of records $n$ and a perfect crowd.}\label{tab:synthtime}
\centering
\begin{tabular}{c|cccc}
\hline
$n$ & 100 & 250 & 500 & 1000\\
\hline
{\bf Average} & 5.81ms & 14.01ms & 40.3ms & 130.4ms\\
{\bf Max} & 36.33ms & 24.4ms & 58.7ms & 172.5ms\\
{\bf Min} & 4.4ms & 11.55ms & 34.88ms & 118.18ms\\
\hline
\end{tabular}
\end{table}

\subsection{Complexity}
In practice, the {\sc Update} as well as {\sc Resolve} algorithm perform near linear when scaling up the number of nodes in the votes graph.
To demonstrate the scaling behavior, we examined the processing time, i.e.,~the end-to-end time spent on integrating an update, for a synthetic dataset as shown in \Autoref{tab:synthtime}.
The dataset has a maximum entity size of 50 and the distribution of entities within the ground truth follows a Zipf distribution.
Varying the number of records, we can see that the processing time increases as expected due to the bigger votes graph which entails more propagation of updates.
The increase is not quadratic but approximately linear:
For $n=100$ the time spent on average per update is 5.81ms, for $n=250$ 14.01ms, for $n=500$ 40.3ms, and for $n=1000$ 130.4ms.  
The reason is that the graph size is not the only influence on execution time.
Entity size as well as the amount of noise in the votes graph play a decisive role for our algorithm's performance.
The higher the number of positive votes in the graph, the higher the number of positive or negative paths in the votes graph.
As a result, processing time per update will increase.

In different synthetic experiments we observe that the computational effort per update is not the bottleneck for crowdsourced ER.
Instead, experiments where clusters are very small usually take more time to execute end to end.
The reason is that even though the time spent on one update is small, a lot more updates are required to complete the votes graph and to find a path between all records.
The processing time of our real-world datasets is discussed in \Autoref{sec:experiments}.

\subsection{Path Optimization}\label{sec:faultpathopt}
While constructing the {\sc MinMax} matrix, the {\sc Update} algorithm described above maintains the current path and its score.
To enable task requests efficiently, this knowledge can be used to define the investment points for further requests:
If the score of a path needs to be modified (i.e.,~to increase certainty in the result), the minimal edge(s) in that path are those record pairs that have to be answered by the crowd.
We observe that a minimal edge in a path between records $r_i$ and $r_j$ can most likely be improved if it satisfies either of the following requirements:
\begin{itemize}
  \item If the path that is currently inspected is positive, then $p_{ij}\geq n_{ij}$ must hold.
  \item If the path is negative, then $n_{ij}\geq p_{ij}$ must hold.
\end{itemize}
The intuition here is that an edge can only get stronger if it gets more votes in its favor.
This is more likely if it is already a dominating edge because if it is dominated, the crowd so far thought this decision to be the wrong one for this record relationship.
Asking more questions for that record pair will thus likely result in the strengthening of the opposite decision which will not enhance the certainty of the candidate pair in question.
\begin{example}[Edge Selection]
Observing the paths connecting records $r_2$ and $r_3$ in \Autoref{fig:MMexample}, there are two possible minimal paths, $\{r_3,r_1,r_2\}$ and $\{r_3,r_4,r_2\}$.
The first path has its minimal edge between records $r_1$ and $r_3$ and since this edge is not dominated, improving the path at minimal cost would result in a request for more information on that record pair.
The improvement of the second path on the other is not feasible as the minimal edge between $r_2$ and $r_4$ is dominated by its counterpart.
\end{example}

\section{Next-Crowdsource Problem}\label{sec:next}
As explained in Section \ref{sec:sub:problem}, the next-crowdsource problem is to decide which record pairs should be resolved next to improve result quality for the lowest possible additional budget. 
Whichever task, i.e.,~comparison between records, is asked next is dependent on the current state of the votes graph.
If there is no knowledge on how any records are related, there cannot be an informed decision on which task to issue to the crowd.
On the other hand, if there is some pairwise information available it can be leveraged to refine the current ER solution.

In our work, we assume independence of data preprocessing but can leverage a preprocessing step analogous to the steps described in \cite{DBLP:conf/sigmod/WangLKFF13,DBLP:journals/pvldb/WhangLG13}.
Please refer to the publications dataset in \Autoref{sec:experiments} as an example.
Note however that preprocessing steps such as hint generation through automatic similarity computation or blocking mechanisms that reduce the search space are an orthogonal problem to this line of work.
The reason is that automatic similarity metrics commonly rely on some kind of syntactic similarity that is not necessarily equivalent to what humans perceive.
For example if pictures show a well-known landmark from different angles, humans can use knowledge about the building to answer the task based on semantic knowledge.
Thus, we focus on enhancing the already existing votes graph through human knowledge in this section.

\smallsection{Solution Space}
Two alternative approaches to solving the next-crowdsource problem through a queuing system are discussed in \Autoref{sec:sub:problem}:
The first one generates all candidate pairs, monotonically requesting pairs, and terminating when all pairs have been asked or inferred.
The second approach generates candidate pairs, requesting them iteratively while regularly checking whether other record pairs should be reinserted.
Thus, it proceeds non-monotonically.
For each of these approaches, a specified decision function provides information about record relationships.
If there exist uncertain record pairs, the next-crowdsource component needs to decide a) in which order uncertain pairs should be addressed and b) how much budget should be allocated per candidate pair.
Furthermore, when working with the crowd execution time is always a factor as outsourcing tasks to humans lowers the overall end-to-end time especially if the applied algorithm proceeds sequentially.
This makes parallelization techniques and how they impact result quality and execution time an important alley to explore.
To address both of these challenges, the next section details three querying strategies that can be used for record pair ordering.
We introduce two types of parallelism for incremental ER and discuss their benefits and disadvantages in comparison to sequential task execution.

%
%
%
%
%
%
%

\subsection{Querying Strategies}
We call {\it querying strategies} those budget allocation strategies that determine the ordering of unresolved candidate pairs in the priority queue.
To determine the ordering of candidate pairs, they use the consensus measure $\varphi_{ij}$ of two records $r_i$ and $r_j$.
The consensus of a record pair [$r_i$,$r_j$] describes how a decision function $f$ estimates the relationship of $r_i$ and $r_j$ numerically.
More specifically, it is a potentially enriched output of $f$.
In \Autoref{sec:faulttolerant}, we discussed decision funtions that map the relationship of $r_i$ and $r_j$ to `Yes', `No', and `Unknown' decisions.
If $\varphi_{ij} \in [-1,1]$, these correspond to 1 (`Yes'), -1 (`No'), and 0 (`Unknown').
Given more expressive decision functions such as $f_M$ or majority-based decision mechanisms, the consensus measure can be adapted:
For example given a quorum of 3, $p^*_{ij}=3$, and $n^*_{ij}=1$, $\varphi_{ij}$ is 0.67.
Other decision functions that can be employed here are normalized distance-based metrics or adjacency-based mechanisms that decide on the consensus based on the relationship of a record to its neighbors such as \cite{DBLP:journals/pvldb/VesdapuntBD14}.
Querying strategies then use the consensus of record pairs to prioritize record pairs over others.
\begin{definition}[Querying Strategy]
The querying strategy $\omega$ of an ER process determines the order in which candidate pairs are evaluated.
For that purpose, it uses the consensus measure $\varphi_{ij}$ to determine the relative position of a record pair [$r_i$,$r_j$] in $Q$.
\end{definition}
Any querying strategy applied in the context of ER realizes some trade-off between the budget spent and the quality that it wants to achieve.
However, strategies may vary on the exact trade-off and whether they favor quality over cost or vice versa.
In general, querying strategies are similar to active learning \cite{DBLP:series/synthesis/2012Settles} because we try to learn the results by actively asking more questions.

The first strategy that we present next aims to maintain consistency within the part of the dataset that the algorithm has already looked at.
It is therefore designed to verify that the decisions made so far are in fact correct.
We refer to this querying strategy as the {\it error reduction} strategy ({\sc ErS}).
\begin{definition}[ErS]
Given record pairs [$r_i$,$r_j$] and [$r_k$,$r_l$], query strategy $\omega_{\sc ErS}$ prioritizes [$r_i$,$r_j$] over [$r_k$,$r_l$] if
$$|\varphi(r_i,r_j)| > |\varphi(r_k,r_l)|$$
and if for both pairs of records $|\varphi(r_i,r_j)|\neq~1$ respectively $\varphi(r_k,r_l)|\neq~1$ holds.
\end{definition}
This strategy is able to provide high precision ER solutions:
Instead of investing the budget into solving a lot of different pairs only partially, it focuses on resolving edges before moving on to the next recor pair.
As a result, this strategy provides introduces votes into the graph that it is certain about.
In contrast, the second strategy that we introduce promotes completeness over correctness and distributes the budget in a breadth-first manner.
It allows for a fast (but possibly imprecise) initial representation of record relationships within the dataset.
We will refer to this strategy as the {\it uncertainty reduction} strategy ({\sc UrS}) as it tries to obtain some information for any record pair first before asking more questions in depth to become more certain the record relationships.
\begin{definition}[UrS]
Given record pairs [$r_i$,$r_j$] and [$r_k$,$r_l$], query strategy $\omega_{\sc UrS}$ prioritizes [$r_i$,$r_j$] over [$r_k$,$r_l$] if
$$|\varphi(r_i,r_j)| < |\varphi(r_k,r_l)|$$
and if for both pairs of records $|\varphi(r_i,r_j)|\neq~1$ respectively $\varphi(r_k,r_l)|\neq~1$ holds.
\end{definition}
As we assume uncertainty to be at its highest when $\varphi(r_i,r_j)=0$, the absolute distance of the consensus measure to 0 is the certainty of [$r_i$,$r_j$] for $f$.
With no prior knowledge of record pairs, [$r_i$,$r_j$]$=0$ holds.
Record pair [$r_i$,$r_j$] would therefore be prioritized in comparison to any other record pair [$r_k$,$r_l$] for which $0<|\varphi(r_k,r_l)|<1$ holds.
The reason is that such a score would indicate some knowledge of [$r_k$,$r_l$].

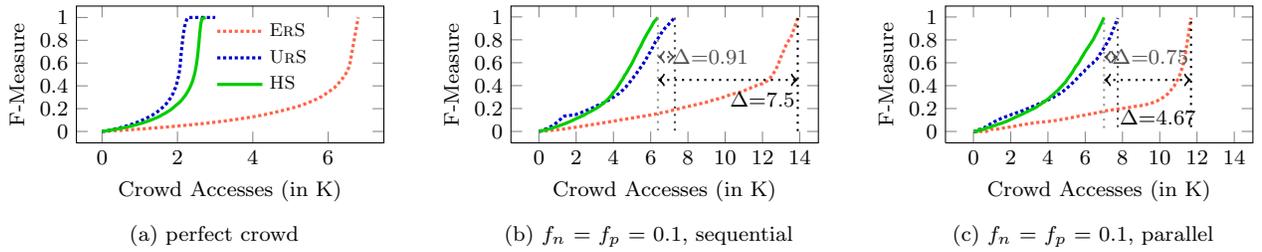
\begin{figure*}[t]
\centering
\pgfplotstableread{graphdata/landmarks_strategycomparison_sequential_fp0fn0.dat}{\strategiesperfect}
\pgfplotstableread{graphdata/landmarks_strategycomparison_sequential_fp01fn01.dat}{\strategiesnoise}
\pgfplotstableread{graphdata/landmarks_strategycomparison_parallel_fp01fn01.dat}{\strategiesnoiseparallel}
\begin{subfigure}[b]{.32\textwidth}
\begin{tikzpicture}[font=\small]
	\small
	\begin{axis}
	[xlabel=Crowd Accesses (in K),
	ylabel=F-Measure,
	legend style={at={(axis cs:3,1)},anchor=north west, nodes=right, draw=none, inner sep = 0pt, font=\scriptsize},
	xticklabel style={/pgf/number format/.cd,fixed,precision=2},
    xtick={0,2,4,6,8},
    ytick={0,0.2,0.4,0.6,0.8,1.0},
    width=\columnwidth,
    height=.6\columnwidth
    ]
	\addplot [red1, densely dotted,very thick, mark options=solid] table [x={cost_cau}, y={cau}] {\strategiesperfect};
	\addlegendentry{{\sc ErS}}
	\addplot [blue1, densely dotted,very thick, mark options=solid] table [x={cost_opt}, y={opt}] {\strategiesperfect};
	\addlegendentry{{\sc UrS}}
	\addplot [green2,very thick, mark options=solid] table [x={cost_hyb}, y={hyb}] {\strategiesperfect};
	\addlegendentry{{\sc HS}}
	\end{axis}%
\end{tikzpicture}
\caption{perfect crowd}\label{fig:landmarkssynth0}
\end{subfigure}
\begin{subfigure}[b]{.32\textwidth}
\begin{tikzpicture}[font=\small]
	\small
	\begin{axis}
	[xlabel=Crowd Accesses (in K),
	ylabel=F-Measure,
	xticklabel style={/pgf/number format/.cd,fixed,precision=2},
    xtick={0,2,4,6,8,10,12,14},
    xmax=15,
    ytick={0,0.2,0.4,0.6,0.8,1.0},
    width=\columnwidth,
    height=.6\columnwidth
    ]
	\addplot [red1, densely dotted,very thick, mark options=solid] table [x={cost_cau}, y={cau}] {\strategiesnoise};
	\addplot [blue1, densely dotted,very thick, mark options=solid] table [x={cost_opt}, y={opt}] {\strategiesnoise};
	\addplot [green2,very thick, mark options=solid] table [x={cost_hyb}, y={hyb}] {\strategiesnoise};

	\draw[black!40, thick, dotted] (axis cs:6.389,0.995) -- (axis cs:6.389,0);
	
	\draw[black, thick, dotted] (axis cs:13.892,0.990) -- (axis cs:13.892,0);
	\draw[<->, black, thick, dotted] (axis cs:6.389,0.45) -- (axis cs:13.892,0.45);
	\node[black] (d) at (axis cs: 12,0.28) {\small$\Delta$=7.5};
	
	\draw[black!70, thick, dotted] (axis cs:7.297,0.994) -- (axis cs:7.297,0);
	\draw[<->, black!70, thick, dotted] (axis cs:6.389,0.65) -- (axis cs:7.297,0.65);
	\node[black!70] (d) at (axis cs: 9.2,.65) {\small$\Delta$=0.91};
	
	\end{axis}%
\end{tikzpicture}
\caption{$f_n$ = $f_p$ = 0.1, sequential}\label{fig:landmarkssynth01}
\end{subfigure}
\begin{subfigure}[b]{.32\textwidth}
\begin{tikzpicture}[font=\small]
	\small
	\begin{axis}
	[xlabel=Crowd Accesses (in K),
	ylabel=F-Measure,
	legend style={at={(axis cs:16,.8)},anchor=north west, nodes=right, draw=none, inner sep = 0pt, font=\scriptsize},
	xticklabel style={/pgf/number format/.cd,fixed,precision=2},
    xtick={0,2,4,6,8,10,12,14},
    xmax=15,
    ytick={0,0.2,0.4,0.6,0.8,1.0},
    width=\columnwidth,
    height=.6\columnwidth
    ]
	\addplot [red1, densely dotted,very thick, mark options=solid] table [x={cost_cau}, y={cau}] {\strategiesnoiseparallel};
	\addplot [blue1, densely dotted,very thick, mark options=solid] table [x={cost_opt}, y={opt}] {\strategiesnoiseparallel};
	\addplot [green2,very thick, mark options=solid] table [x={cost_hyb}, y={hyb}] {\strategiesnoiseparallel};
	
	\draw[black!40, thick, dotted] (axis cs:7.01,0.9529) -- (axis cs:7.01,0);
	
	\draw[black, thick, dotted] (axis cs:11.678,0.9606) -- (axis cs:11.678,0);
	\draw[<->, black, thick, dotted] (axis cs:7.01,0.45) -- (axis cs:11.678,0.45);
	\node[black] (d) at (axis cs: 9.9,0.12) {\small$\Delta$=4.67};
	
	\draw[black!70, thick, dotted] (axis cs:7.748,0.9606) -- (axis cs:7.748,0);
	\draw[<->, black!70, thick, dotted] (axis cs:7.01,0.65) -- (axis cs:7.748,0.65);
	\node[black!70] (d) at (axis cs: 9.5,.65) {\small$\Delta$=0.75};
	\end{axis}%
\end{tikzpicture}
\caption{$f_n$ = $f_p$ = 0.1, parallel}\label{fig:landmarkssynth01par}
\end{subfigure}
\caption{Landmarks dataset, with noise through varying synthetic crowds for varying querying strategies.}\label{fig:perfectcrowd}
\end{figure*}

\smallsection{Hybrid strategy ({\sc HS})}
Under the assumption that the crowd fulfills their tasks perfectly, {\sc UrS} will always dominate {\sc ErS} in terms of its quality gain per cost unit as shown in \Autoref{fig:landmarkssynth0} on a dataset consisting of landmarks in Paris, France, and Barcelona, Spain with a perfect synthetic oracle.
Quality is captured here through the f-measure of the ER solution while one cost unit corresponds to one evaluated record pair, i.e.,~one crowdsoured task.
Details about the dataset and the experimental setup can be found in \Autoref{sec:experiments}.

In real-world use cases, the assumption of a perfect oracle does not hold which causes the dynamics of the two querying strategies to shift:
The quality improvement of {\sc ErS} is predictable albeit at a higher overall cost while {\sc UrS} makes an increased number of false local decisions especially if only little budget has been invested.
These issues can be addressed through a hybrid querying strategy that combines the best of both.
As positive task decisions have more impact on the ER solution due to the way (anti-)transitivity is defined, the hybrid strategy resolves positive decision candidates cautiously, following {\sc ErS}.
Candidate pairs that potentially lead to negative decisions are processed according to {\sc UrS} to avoid spending budget on a request that potentially has not a lot of impact.
This differentiation of decision strategies is analogous to the idea of reward functions in active learning:
A positive vote indicates that the task is an important one and should be prioritized.
In contrast, a negative vote would not trigger a reward and the record pair would be crowdsourced again at a later point again.
\begin{definition}[Hybrid Strategy]
Given record pairs [$r_i$,$r_j$] and [$r_k$,$r_l$], query strategy $\omega_{\sc HS}$ prioritizes [$r_i$,$r_j$] over [$r_k$,$r_l$] if
$$\varphi(r_i,r_j) > \varphi(r_k,r_l)$$
and if for both pairs of records $|\varphi(r_i,r_j)|\neq~1$ respectively $\varphi(r_k,r_l)|\neq~1$ holds.
\end{definition}
If $\varphi(r_i,r_j)$ and $\varphi(r_k,r_l)$ are positive, the strategy will thus choose the record pair that has less uncertainty but is not yet certain following $\omega_{ErS}$.
In contrast, if $\varphi(r_i,r_j)$ is below 0, it will only be chosen over $\varphi(r_k,r_l)$ if its distance to 0 is lower which is the case if its consensus measure is bigger.

To enable a seamless hybrid strategy, it can be implemented through a double queue system:
One of the queues contains positive candidate pairs and applies {\sc ErS} while {\sc UrS} is used for the negative queue.
Whenever a record pair $[r_i,r_j]$ is inserted into this hybrid queuing system, it is inserted into either the positive or negative queue based on the current state of $p^*_{ij}$ and $n^*_{ij}$.
For record pair extraction, the positive queue is accessed first and a record pair is pulled from the negative queue only if the positive queue does not contain any pairs.

\begin{figure}[t]
\begin{minipage}{.5\columnwidth}
	\centering
	\includegraphics[width=.52\columnwidth]{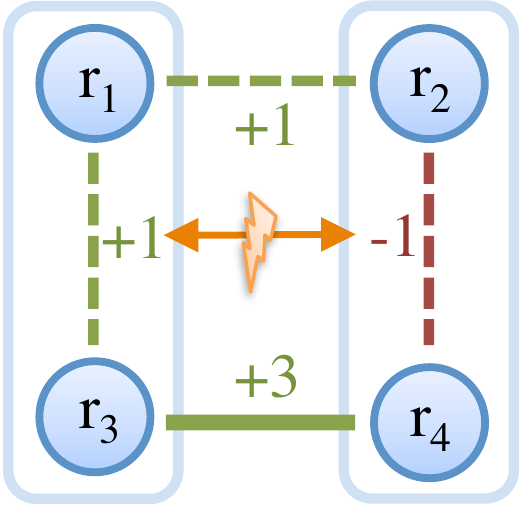}
\end{minipage}
\hfill
\begin{minipage}[t]{.48\columnwidth}
\begin{tabular}{c|c c c c}
\multicolumn{1}{c|}{}& {$r_1$} & {$r_2$} & {$r_3$} & {$r_4$}\\
\hline
$r_1$ & - & \color{ForestGreen}{+1} & \color{ForestGreen}{+1} & \color{ForestGreen}{+1}\\
$r_2$ & \color{BrickRed}{-1} & - & \color{ForestGreen}{+1} & \color{ForestGreen}{+1}\\
$r_3$ & \color{BrickRed}{-1} & \color{BrickRed}{-1} & - & \color{ForestGreen}{+3}\\
$r_4$ & \color{BrickRed}{-1} & \color{BrickRed}{-1} & \color{BrickRed}{-1} & -\\
\end{tabular}
\end{minipage}

\vspace{.07in}
\hspace{.55in}\small{Votes Graph \hspace{.95in} {\sc MinMax} Matrix}

\caption{Parallelization problem.}\label{fig:parallelproblem}
\end{figure}

\subsection{Workload Parallelization}\label{subsec:parallel}
Integrating the crowd into the information collection process comes not only at a monetary but also temporal cost.
In fact, it increases response time drastically.
This is why task parallelization has become an established technique to limit the increase in execution time.
Generally, there are two types of parallelization that are applicable for this kind of pair-wise crowdsourcing, {\it intra-task} and {\it inter-task} parallelization.

\smallsection{Intra-task parallelization}
Instead of asking for a record pair once, the same pair is issued multiple times to different workers.
For example for quorum or majority-based techniques, it is possible to compute the required number of tasks to achieve certainty on the fly:
It depends on how many answers are required and how many tasks have been returned for this record pair previously.

\smallsection{Inter-task parallelization}
Instead of asking for a series of distinct record pairs sequentially, the required pairs are analyzed and if they found independent of each other, they are released onto the crowdsourcing platform in parallel within a batch.
This technique has been discussed in prior work, \cite{DBLP:conf/sigmod/WangLKFF13}, and is commonly realized by generating a spanning tree over all entities.
In practice, we point out that uncertainty of the crowd hinders inter-task parallelization as shown with the following example.
\begin{example}[Inter-task Parallelization]
In \Autoref{fig:parallelproblem}, [$r_3$,$r_4$] is already known after which the algorithm decides to parallelize tasks [$r_1$,$r_2$], [$r_1$,$r_3$], and [$r_2$,$r_4$].
The answers to these tasks causes uncertainty to arise between $r_1$ and $r_3$ resp. $r_2$ and $r_4$.
\end{example}
When comparing sequential and parallel task execution for the different querying strategies with noisy answers occuring with a likelihood of 10\%, (\Autoref{fig:landmarkssynth01} respectively \Autoref{fig:landmarkssynth01par}), the following three observations can be made:
First, the absolute cost for {\sc ErS} decreases when parallelizing task execution.
The reason herefore is that inter-task parallelization requires connecting previously unconnected entities in a manner aimed to provide maximal coverage.
To generate the spanning tree, the algorithm iterates over all candidate pairs in the queue in descending order of their objective.
As a result, it is more likely that positive rather than negative candidate pairs are part of the spanning tree which in return improves the performance of {\sc ErS}.
The second observation is that the cost for both {\sc UrS} and {\sc HS} increases.
Here, the opposite effect as for {\sc ErS} takes hold:
Parallelizing them forces both strategies to execute tasks with lower priority within the current batch if other (higher priority) tasks are not consistent with the existing spanning tree.
Additionally, the resolution of contradictions in the result set incurs higher cost because of the necessity for more information but at the same time it lowers result quality due to noisier intermediate results.
Third, {\sc HS} maintains a better result cost and quality trade-off than either alternative approach because it is still able to leverage the benefits of both {\sc ErS} and {\sc UrS}.

\section{Experimental Evaluation}\label{sec:experiments}

In this section, we compare fault-tolerant data interpretation mechanisms with consensus-based approaches and discuss the different querying strategies and parallelization techniques introduced in the previous section.

\subsection{Experimental Setup}
Two different real-world datasets that contain pictures resp. publication records are used in this evaluation.
We evaluated these datasets both with synthetic and real-world crowds.
To provide an extensive experimental evaluation on crowdsourced data, the datasets were crowdsourced completely on the Amazon Mechanical Turk platform, i.e.,~ten crowd workers solved each possible record pair.
Each of these workers has a 90\% acceptance rate to avoid malicious workers.
For all the experiments shown here, the results of the respective setup are averaged over at least 100 different single experiments.
Note that this methodology allows us to compare the algorithms in a robust setting where all of them have the same chance to succeed.

\smallsection{Landmarks dataset}
This dataset consists of 266 pictures of landmarks in two European cities, namely Paris, France, and Barcelona, Spain.
It is based on a picture classification dataset for visual object identification algorithms, \cite{Avrithis2010}, and contains 13 entities that are landmarks such as the Arc de Triomphe.
In each task, a crowd worker is shown a pair of pictures and decides whether they show the same landmark.
As shown in Table \ref{tab:datasets}, there exist a total of 352,450 record pairs out of which about 7.8\% are correctly identified as belonging to the same landmark. 
The crowd also identified another 4.6\% as positive matches even though the ground truth was negative.
As a result, 37.1\% of all positive answers should be in fact negative.
In total, about 2.8\% of all extracted answers are false negative answers and 84.8\% of all answers were true negatives.

\smallsection{Publications dataset}
This dataset is a compressed version of the {\sc Cora} dataset which is commonly used for string similarity evaluation \footnote{http://secondstring.sourceforge.net/}.
It is derived from the original dataset through relative compression, i.e.,~if an entity contained 5\% of the records in the original dataset, it will contain approximately the same percentage of randomly selected records in the smaller dataset.
Overall, this dataset has 3.47 records per entity on average and thus results in a higher amount of candidate negative decisions than the landmarks dataset.
In fact, only 2.42\% of all decision pairs are positive, out of which 64.04\% are identified as such.
We observe a high ratio of false positive decisions to the overall number of positive decisions:
Here, 36.2\% of all positive decisions are erroneous.
We also observe that the absolute number of false positive and false negative decisions in the answer set are about equal and overall lower than for the landmarks dataset.

{
\small
\begin{table}[t]
\centering
\resizebox{\columnwidth}{!}{
\begin{tabular}{|c|l|c|c|c|c|}
\cline{3-6}
\multicolumn{2}{c|}{}& \multicolumn{2}{c|}{\bf Landmarks} & \multicolumn{2}{c|}{\bf Publications}\\
\hline
 & \#records & \multicolumn{2}{c|}{266} & \multicolumn{2}{c|}{198} \\
{\sc Data} & \#rec. pairs & \multicolumn{2}{c|}{35,245} & \multicolumn{2}{c|}{19,503} \\
{\sc Statis-} & \#entities & \multicolumn{2}{c|}{13} & \multicolumn{2}{c|}{57} \\
{\sc tics} & avg \#rec./ent. & \multicolumn{2}{c|}{20.46} & \multicolumn{2}{c|}{3.47}\\
 & max \#rec./ent. & \multicolumn{2}{c|}{43} & \multicolumn{2}{c|}{14} \\
 & min \#rec./ent. & \multicolumn{2}{c|}{7} & \multicolumn{2}{c|}{1} \\
\hline
\multicolumn{2}{c|}{}& absolute & {$\%$} & absolute & {$\%$}\\
\hline
{\sc Crowd} & \#true positives & 27,512 & 7.8 & 3,023 & 1.6  \\
{\sc Statis-} & \#true negatives & 298,822 & 84.8 & 188,592 & 96.7 \\
{\sc tics} & \#false positives & 16,248 & 4.6 & 1,718 & 0.9\\
 & \#false negatives & 9,868 & 2.8 & 1,697 & 0.9\\
\hline
\end{tabular}
}
\caption{Datasets Overview}\label{tab:datasets}
\end{table}
}

\begin{figure*}[t]
\centering
\pgfplotstableread{graphdata/landmarks_approaches_sequential_fp03fn0.dat}{\approachesfp}
\pgfplotstableread{graphdata/landmarks_approaches_sequential_fp0fn03.dat}{\approachesfn}
\pgfplotstableread{graphdata/landmarks_approaches_sequential_fp03fn03.dat}{\approachesfnfp}
\begin{subfigure}[b]{.32\textwidth}
\begin{tikzpicture}[font=\small]
	\small
	\begin{axis}
	[xlabel=Crowd Accesses (in K),
	ylabel=F-Measure,
	legend style={at={(axis cs:25,.75)},anchor=north west, nodes=right, draw=none, inner sep = 0pt, outer sep = 0pt, font=\tiny},
	xticklabel style={/pgf/number format/.cd,fixed,precision=2},
    xtick={0,10,20,30,40,50,60},
    ytick={0,0.2,0.4,0.6,0.8,1.0},
    width=\columnwidth,
    height=.6\columnwidth
    ]
	\addplot [green2,very thick, mark options=solid] table [x={cost_feer}, y={feer}] {\approachesfp};
	\addplot [green1,very thick, mark options=solid] table [x={cost_fer}, y={fer}] {\approachesfp};
	\addplot [red1, very thick, mark options=solid] table [x={cost_cer5}, y={cer5}] {\approachesfp};
	\addplot [red1, densely dotted,very thick, mark options=solid] table [x={cost_cer9}, y={cer9}] {\approachesfp};
	\end{axis}%
\end{tikzpicture}
\caption{$f_p$ = 0.3, $f_n$ = 0}\label{fig:landmarkssynthfp03fn0}
\end{subfigure}
\begin{subfigure}[b]{.32\textwidth}
\begin{tikzpicture}[font=\small]
	\small
	\begin{axis}
	[xlabel=Crowd Accesses (in K),
	ylabel=F-Measure,
	xticklabel style={/pgf/number format/.cd,fixed,precision=2},
    xtick={0,10,20,30,40,50},
    ytick={0,0.2,0.4,0.6,0.8,1.0},
    width=\columnwidth,
    height=.6\columnwidth
    ]
	\addplot [green2,very thick, mark options=solid] table [x={cost_feer}, y={feer}] {\approachesfn};
	\addplot [green1,very thick, mark options=solid] table [x={cost_fer}, y={fer}] {\approachesfn};
	\addplot [red1, very thick, mark options=solid] table [x={cost_cer5}, y={cer5}] {\approachesfn};
	\addplot [red1, densely dotted,very thick, mark options=solid] table [x={cost_cer9}, y={cer9}] {\approachesfn};
	
	\end{axis}%
\end{tikzpicture}
\caption{$f_p$ = 0, $f_n$ = 0.3}\label{fig:landmarkssynthfp0fn03}
\end{subfigure}
\begin{subfigure}[b]{.32\textwidth}
\begin{tikzpicture}[font=\small]
	\small
	\begin{axis}
	[xlabel=Crowd Accesses (in K),
	ylabel=F-Measure,
	legend style={at={(axis cs:88,.65)},anchor=north west, nodes=right, draw=none, inner sep = 0pt, font=\tiny},
	xticklabel style={/pgf/number format/.cd,fixed,precision=2},
    xtick={0,40,80,120,160},
    xmax = 190,
    ytick={0,0.2,0.4,0.6,0.8,1.0},
    width=\columnwidth,
    height=.6\columnwidth
    ]
	\addplot [green2,very thick, mark options=solid] table [x={cost_feer}, y={feer}] {\approachesfnfp};
	\addlegendentry{{\sc Feer}}
	\addplot [green1,very thick, mark options=solid] table [x={cost_fer}, y={fer}] {\approachesfnfp};
	\addlegendentry{{\sc Fer}}
	\addplot [red1, very thick, mark options=solid] table [x={cost_cer5}, y={cer5}] {\approachesfnfp};
	\addlegendentry{{\sc Cer},$|v|=5$}
	\addplot [red1, densely dotted,very thick, mark options=solid] table [x={cost_cer9}, y={cer9}] {\approachesfnfp};
	\addlegendentry{{\sc Cer},$|v|=9$}
	
	\end{axis}%
\end{tikzpicture}
\caption{$f_p$ = $f_n$ = 0.3}\label{fig:landmarkssynthfp30fn03}
\end{subfigure}
\caption{Landmarks dataset, with noise through varying synthetic crowds for varying data interpretation models.}\label{fig:synthlandmarks}
\end{figure*}
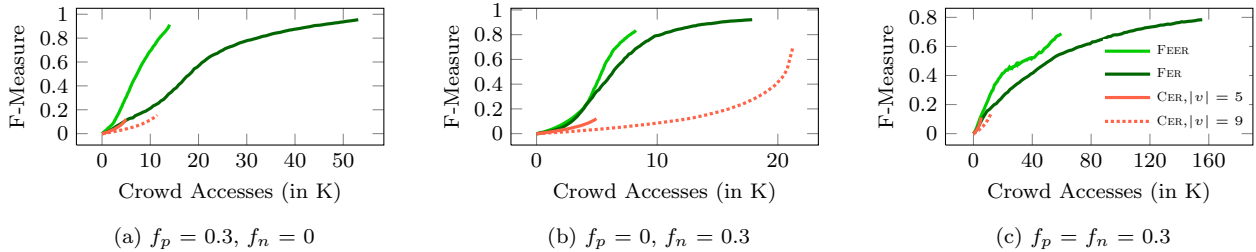

\begin{figure}[t]
\centering
\pgfplotstableread{graphdata/landmarks_approaches_sequential_fp0fn0.dat}{\approachesperf}
\begin{subfigure}[b]{.7\columnwidth}
\begin{tikzpicture}[font=\small]
	\small
	\begin{axis}
	[xlabel=Crowd Accesses (in K),
	ylabel=F-Measure,
	legend style={at={(axis cs:22,.8)},anchor=north west, nodes=right, draw=none, outer sep = 0pt, inner sep = 0pt, font=\tiny},
	xticklabel style={/pgf/number format/.cd,fixed,precision=2},
    xtick={0,10,20,30,40},
    ytick={0,0.2,0.4,0.6,0.8,1.0},
    width=\columnwidth,
    height=.6\columnwidth
    ]
	\addplot [green2,very thick, mark options=solid] table [x={cost_feer}, y={feer}] {\approachesperf};
	\addlegendentry{{\sc Feer}}
	\addplot [green1,very thick, mark options=solid] table [x={cost_fer}, y={fer}] {\approachesperf};
	\addlegendentry{{\sc Fer}}
	\addplot [red1, very thick, mark options=solid] table [x={cost_cer5}, y={cer5}] {\approachesperf};
	\addlegendentry{{\sc Cer}, $|v|=5$}
	\addplot [red1, densely dotted,very thick, mark options=solid] table [x={cost_cer9}, y={cer9}] {\approachesperf};
	\addlegendentry{{\sc Cer}, $|v|=9$}
	
	\end{axis}%
\end{tikzpicture}
\end{subfigure}
\caption{Landmarks dataset, with perfect synthetic crowd for varying data interpretation models.}\label{fig:landmarksperf}
\end{figure}
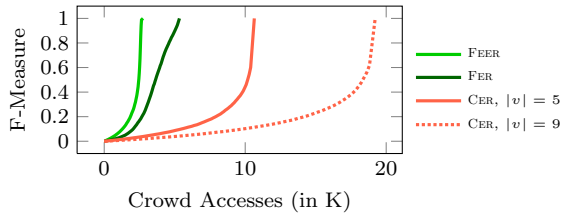

\smallsection{Algorithms}
We evaluate three different data interpretation algorithms from the core classes identified in Section \ref{sec:problem}.

\begin{itemize}
  \item {\sc Cer}. As state-of-the-art consensus-based strategy we choose CrowdER, \cite{DBLP:conf/sigmod/WangLKFF13}, which uses a majority-based decision strategy.
  It monotonically requests all candidate pairs from the crowd, merging those that are connected by a positive decision and can thus be classified as a consensus-based decision strategy.
  \item {\sc Fer}. This algorithm represents the class of fault-tolerant strategies.
  It incorporates the same non-repetitive queuing mechanism as {\sc Cer} but uses {\sc MinMax} as decision function with a quorum $q=3$ and repeats each task until $q$ or the edge budget ($b_E$=10) is reached.
  Experiments with other quorum values are omitted here but show similar trends.
  \item {\sc Feer}. This is an exhaustive fault-tolerant strategy which implements {\sc Fer} but maintains a non-monotonic queuing system:
  If an update to the internal {\sc MinMax} strategy causes a pair to become uncertain, it is inserted back into the queue.
\end{itemize}

This evaluation will also focus on the two task ordering techniques discussed in \Autoref{sec:next}: queuing strategies and parallelization.
We evaluate the three queuing strategies presented previously, {\sc ErS} which is a cautious mechanism that aims to optimize intermediate results, {\sc UrS} which internally orders its candidate pairs according to their level of uncertainty, and their hybrid {\sc HS} which cautiously processes positive and optimistically processes negative candidates.
The hybrid strategy will serve as default strategy if not otherwise declared.
To assess the impact of parallelization techniques, we implement as baseline a sequential process that iteratively asks for exactly one record pair.
We compare it to a combination of inter- and intra-task parallelization, i.e.,~if parallelized, the algorithm automatically computes a spanning tree over all entities and within an entity as well.
For all candidate pairs, it then computes the minimal necessary investment to reach quorum or consensus and issues the resulting record pairs as a batch.

\smallsection{Metrics \& Implementation}
This evaluation uses mainly two standardized metrics, quality and cost.
The cost of any experiment is simply measured as the number crowd accesses which is equivalent to the amount of requested record pairs.
To measure quality, we use precision and recall where precision is the percentage of record pairs that are correctly associated with the same entity and recall is the percentage of record pairs that we correctly assign to the same entity.
To provide a unified quality metric, the f-measure of these values is used as standard quality metric, defined as $\frac{P*R*2}{P+R}$.
We implemented all of the presented algorithms and strategies in Java, and experimented on a Linux machine with eight Intel Xeon L5520 cores (2.26GHz, cache 24MB).

\subsection{The Impact of Crowd Error}\label{subsec:crowderror}
To exemplify the impact of errors made by crowd workers, we synthetically generated two types of noise for the landmarks dataset.
The first noise is false positive noise, $f_p$, which describes the percentage of decisions where crowd workers wrongly classify records $r_i$ and $r_j$ to belong to the same entity when in fact they belong to different entities according to the ground truth.
Analogously, the second type of noise is false negative noise, $f_n$, where $r_i$ and $r_j$ are falsely assigned to the same entity.
\Autoref{fig:synthlandmarks} and \Autoref{fig:landmarksperf} show the results of this set of experiments, comparing different data interpretation models under uniform synthetic noise varying both $f_n$ and $f_p$.
We observe that if the crowd answers perfectly (\Autoref{fig:landmarksperf}), {\sc Fer} and {\sc Feer} show rapid quality improvement per cost unit.
{\sc Feer} outperforms {\sc Fer} because it dynamically adjust the priority queue, pushing potentially positive candidate pairs to the front of the queue and therefore maximizing the information gain per crowd access.
The cost/quality trade-off for {\sc Cer} is predictable and is correlated to the number of votes $v$ requested:
As the crowd answers perfectly, a smaller $v$ means having a comparabily better quality result for a lower budget.
When introducing noise into the experimental setup, we make the following three observations.

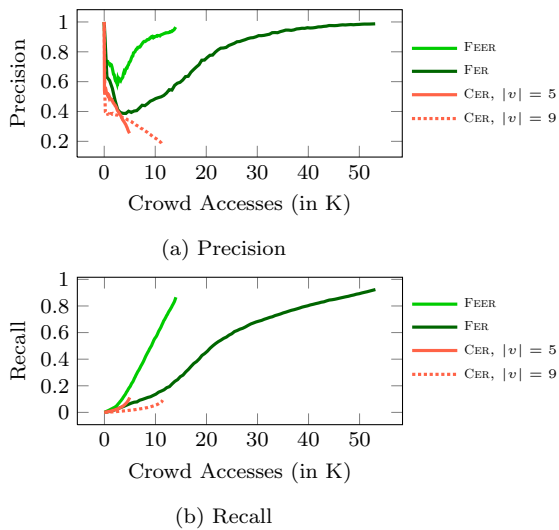
\begin{figure}[t]
\centering
\pgfplotstableread{graphdata/landmarks_approaches_sequential_fp03fn0_precision.dat}{\approachesfpprec}
\pgfplotstableread{graphdata/landmarks_approaches_sequential_fp03fn0_recall.dat}{\approachesfprec}
\begin{subfigure}[b]{.7\columnwidth}
\begin{tikzpicture}[font=\small]
	\small
	\begin{axis}
	[xlabel=Crowd Accesses (in K),
	ylabel=Precision,
	legend style={at={(axis cs:60,.9)},anchor=north west, nodes=right, draw=none, outer sep = 0pt, inner sep = 0pt, font=\tiny},
	xticklabel style={/pgf/number format/.cd,fixed,precision=2},
    xtick={0,10,20,30,40,50,60},
    ytick={0,0.2,0.4,0.6,0.8,1.0},
    width=\columnwidth,
    height=.6\columnwidth
    ]
	\addplot [green2,very thick, mark options=solid] table [x={cost_feer}, y={feer}] {\approachesfpprec};
	\addlegendentry{{\sc Feer}}
	\addplot [green1,very thick, mark options=solid] table [x={cost_fer}, y={fer}] {\approachesfpprec};
	\addlegendentry{{\sc Fer}}
	\addplot [red1, very thick, mark options=solid] table [x={cost_cer5}, y={cer5}] {\approachesfpprec};
	\addlegendentry{{\sc Cer}, $|v|=5$}
	\addplot [red1, densely dotted,very thick, mark options=solid] table [x={cost_cer9}, y={cer9}] {\approachesfpprec};
	\addlegendentry{{\sc Cer}, $|v|=9$}
	
	\end{axis}%
	
\end{tikzpicture}

\caption{Precision}\label{fig:landmarksfpprec}
\end{subfigure}

\begin{subfigure}[b]{.7\columnwidth}
\begin{tikzpicture}[font=\small]
	\small
	\begin{axis}
	[xlabel=Crowd Accesses (in K),
	ylabel=Recall,
	legend style={at={(axis cs:60,.9)},anchor=north west, nodes=right, draw=none, outer sep = 0pt, inner sep = 0pt, font=\tiny},
	xticklabel style={/pgf/number format/.cd,fixed,precision=2},
    xtick={0,10,20,30,40,50},
    ytick={0,0.2,0.4,0.6,0.8,1.0},
    width=\columnwidth,
    height=.6\columnwidth
    ]
	\addplot [green2,very thick, mark options=solid] table [x={cost_feer}, y={feer}] {\approachesfprec};
	\addlegendentry{{\sc Feer}}
	\addplot [green1,very thick, mark options=solid] table [x={cost_fer}, y={fer}] {\approachesfprec};
	\addlegendentry{{\sc Fer}}
	\addplot [red1, very thick, mark options=solid] table [x={cost_cer5}, y={cer5}] {\approachesfprec};
	\addlegendentry{{\sc Cer}, $|v|=5$}
	\addplot [red1, densely dotted,very thick, mark options=solid] table [x={cost_cer9}, y={cer9}] {\approachesfprec};
	\addlegendentry{{\sc Cer}, $|v|=9$}
	
	\end{axis}%
\end{tikzpicture}
\caption{Recall}\label{fig:landmarksfprec}
\end{subfigure}
\caption{Landmarks dataset, with noisy $f_p$=0.3 and $f_n$=0 crowd for varying data interpretation models.}\label{fig:landmarksfpprecrec}
\end{figure}

\smallsection{False Positive Information}
This is the type of crowd error that has the most impact on output quality (\Autoref{fig:landmarkssynthfp03fn0}):
In {\sc Cer}, transitive closure is applied for entity resolution which means that a positive decision leads to a merge of records into one entity.
As these decisions are never questioned, the recall and precision of {\sc Cer} (\Autoref{fig:landmarksfpprecrec}) soon decrease as records are merged that belong to different entities.
Generally, we observe that recall increases over time while precision (i.e.,~records correctly belong to the same entity) decreases over time when generating false positive noise.
{\sc Feer} and {\sc Fer} recover from an initial drop of precision because they question their resolutions whenever they encounter a contradiction.
This can be observed in \Autoref{fig:landmarksfpprec} through the increase in the respective precision curve for {\sc Feer} ({\sc Fer}) after approximately 3k (4k) crowd accesses.

\smallsection{False Negative Information}
If the crowd answers with false negative answers (\Autoref{fig:landmarkssynthfp0fn03}), the workers decide to keep records apart which should be in the same entity with a likelihood of $f_n$ $=$ 0.3.
We observe in this set of experiments that an increase in the number of questions asked ($v$) leads to a provable increase in quality for {\sc Cer}.
Here, with $v=9$ instead of $5$, the final ER solution reaches a f-measure of 0.69 instead of 0.12.
Compared to {\sc Feer} and {\sc Fer}, {\sc Cer} still provides a lower quality improvement over crowd accesses because in order to reach better quality, more budget has to be invested into finding an ER solution.
We notice also in this experiment that varying $f_n$ changes the recall of the results but never the precision.
This behavior is different to varying false positive information which influences both.
To understand the different behaviorisms, remember the definition of precision and recall:
Precision is negatively influenced through records being falsely assigned to the same entity, recall is negatively influenced by records in different entities that should be in the same entity according to the ground truth.
As records can never falsely belong to the same entity through false negative votes, precision is not influenced by this type of crowd error.

\smallsection{Noisy Crowd Information}
Combining both types of crowd error leads to a visible quality decrease for both {\sc Cer} variations and larger budget requirements for both {\sc Feer} and {\sc Fer} (\Autoref{fig:synthlandmarks}).
We again observe that increasing the budget for {\sc Cer} in fact positively influences the result quality but as we have shown for our real-world experiments (\Autoref{sec:experiments}), it never reaches the same level of quality as both fault-tolerant approaches.
The additionally required budget investment for these can be explained through the decision behavior of {\sc MinMax}:
If uncertain, it requests more information from the crowd, thus allocating more budget.
If the budget is granted, it leads to a steady improvement in result quality which itself is better than the {\sc Cer} result quality at any point in time.

\subsection{Landmarks Results}
The landmarks dataset is an interesting use case as it provides an environment where no similarity metric enhances the performance of the ER strategies.
Here, every candidate pair is initially equally likely.
To contrast our techniques, we now compare the fault-tolerant mechanisms with the traditional consensus-based approach using the same crowdsourced results obtained by posting tasks on Amazon Mechanical Turk for all approaches.
Additionally, we evaluate the different queuing strategies after which we highlight the advantages and disadvantages of parallelization.

\subsubsection{Data interpretation}
To compare the quality of fault-tolerant mechanisms to algorithms assuming a perfect world, the first set of experiments compares the cost/quality trade-off for {\sc Cer}, {\sc Fer}, and {\sc Feer} (\Autoref{fig:landmarksseq}).

\smallsection{Performance of {\sc Cer}}
\Autoref{fig:landmarksseqdata} shows the development of all data interpretation mechanisms in terms of quality over crowd access.
The low performance of {\sc Cer} in comparison to {\sc Feer} and {\sc Fer} (see difference in quality $\Delta_q$ and cost $\Delta_c$ in \Autoref{fig:landmarksseqdata}) is due to its sensitivity to false negative decisions which form 26.4\% of all candidate positive decisions in the landmarks dataset.
While {\sc Cer} maintains a minimum of 0.96 precision over time, its recall is at most 0.21 even if the number of votes considered $|v|$ is increased from 5 to 9 which only results in a decrease in cost/quality gain.
At this point, recall that a) positive decisions reduce the search space of the algorithm and b) the landmarks dataset has only 13 big entities.
Uncontested false negative decisions therefore have significant impact on result quality.

\begin{figure}[t]
\centering
\pgfplotstableread{graphdata/landmarks_approaches_sequential.dat}{\approachsequential}
\pgfplotstableread{graphdata/landmarks_connectratio_sequential.dat}{\connectsequential}
\pgfplotstableread{graphdata/landmarks_strategycomparison_sequential.dat}{\strategysequential}
\begin{subfigure}[b]{.7\columnwidth}
\begin{tikzpicture}[font=\small]
	\small
	\begin{axis}
	[xlabel=Crowd Accesses (in K),
	ylabel=F-Measure,
	legend style={at={(axis cs:40,.8)},anchor=north west, nodes=right, draw=none, outer sep = 0pt, inner sep = 0pt, font=\tiny, scale=.5},
	xticklabel style={/pgf/number format/.cd,fixed,precision=2},
    xtick={0,10,20,30,40,50},
    ytick={0,0.2,0.4,0.6,0.8,1.0},
    width=\columnwidth,
    height=.6\columnwidth
    ]
	\addplot [green2,very thick, mark options=solid] table [x={cost_feer}, y={feer}] {\approachsequential};
	\addlegendentry{{\sc Feer}}
	\addplot [green1,very thick, mark options=solid] table [x={cost_fer}, y={fer}] {\approachsequential};
	\addlegendentry{{\sc Fer}}
	\addplot [red1, very thick, mark options=solid] table [x={cost_cer5}, y={cer5}] {\approachsequential};
	\addlegendentry{{\sc Cer},$|v|=5$}
	\addplot [red1, densely dotted,very thick, mark options=solid] table [x={cost_cer9}, y={cer9}] {\approachsequential};
	\addlegendentry{{\sc Cer},$|v|=9$}
	
	\draw[<-, black!70, thick, dotted] (axis cs:27.401,0.8559162808) -- (axis cs:19.535,0.3421579722);
	\draw[<-, black, thick, dotted] (axis cs:10.989,0.7381462161) -- (axis cs:19.535,0.3421579722);
	
	\node[black!70] (d) at (axis cs: 28.5,.58) {\scriptsize$\Delta_c$=7.9};
	\node[black!70] (d) at (axis cs: 29.1,.43) {\scriptsize $\Delta_q$=0.51};
	\node[black] (d) at (axis cs: 3.9,.8) {\scriptsize$\Delta_c$=-8.6};
	\node[black] (d) at (axis cs: 3.4,.65) {\scriptsize $\Delta_q$=0.4};
	
	\end{axis}%
\end{tikzpicture}
\caption{Data Interpretation}\label{fig:landmarksseqdata}
\end{subfigure}

\begin{subfigure}[b]{.7\columnwidth}
\begin{tikzpicture}[font=\small]
	\small
	\begin{axis}
	[xlabel=Crowd Accesses (in K),
	ylabel=F-Measure,
	legend style={at={(axis cs:15.8,.61)},anchor=north west, nodes=right, draw=none, outer sep = -2pt, inner sep = -2pt, font=\tiny},
	xticklabel style={/pgf/number format/.cd,fixed,precision=2},
    xtick={0,10,20,30,40,50},
    ytick={0,0.2,0.4,0.6,0.8,1.0},
    width=\columnwidth,
    height=.6\columnwidth
    ]
    \addplot [green1,very thick, mark options=solid] table [x={cost_fer}, y={fer}] {\approachsequential};
	\addlegendentry{{\sc Fer}}
	\addplot [green2,very thick, mark options=solid] table [x={cost_cr0}, y={cr0}] {\connectsequential};
	\addlegendentry{{\sc Feer}, $\kappa$=0}
	\addplot [green2, densely dotted, very thick, mark options=solid] table [x={cost_cr01}, y={cr01}] {\connectsequential};
	\addlegendentry{{\sc Feer}, $\kappa$=0.1}
	\addplot [green2, dashed, very thick, mark options=solid] table [x={cost_cr02}, y={cr02}] {\connectsequential};
	\addlegendentry{{\sc Feer}, $\kappa$=0.2}
	\end{axis}%
\end{tikzpicture}
\caption{Connectivity Ratios}\label{fig:landmarksseqconnect}
\end{subfigure}

\begin{subfigure}[b]{.7\columnwidth}
\begin{tikzpicture}[font=\small]
	\small
	\begin{axis}
	[xlabel=Crowd Accesses (in K),
	ylabel=F-Measure,
	legend style={at={(axis cs:25,.7)},anchor=north west, nodes=right, draw=none, outer sep = 0pt, inner sep = 0pt, font=\tiny},
	xticklabel style={/pgf/number format/.cd,fixed,precision=2},
    xtick={0,5,10,15,20,25,30,35,40,45,50},
    ytick={0,0.2,0.4,0.6,0.8,1.0},
    width=\columnwidth,
    height=.6\columnwidth
    ]
	\addplot [red1, densely dotted,very thick, mark options=solid] table [x={cost_cau}, y={cau}] {\strategysequential};
	\addlegendentry{{\sc ErS}}
	\addplot [blue1, densely dotted,very thick, mark options=solid] table [x={cost_opt}, y={opt}] {\strategysequential};
	\addlegendentry{{\sc UrS}}
	\addplot [green2,very thick, mark options=solid] table [x={cost_hyb}, y={hyb}] {\strategysequential};
	\addlegendentry{{\sc HS}}
	\end{axis}%
\end{tikzpicture}
\caption{Querying Strategies}\label{fig:landmarksseqqueue}
\end{subfigure}
\caption{Landmarks dataset, sequential mode with varying data interpretation models, connectivity ratios, and querying strategies.}\label{fig:landmarksseq}
\end{figure}
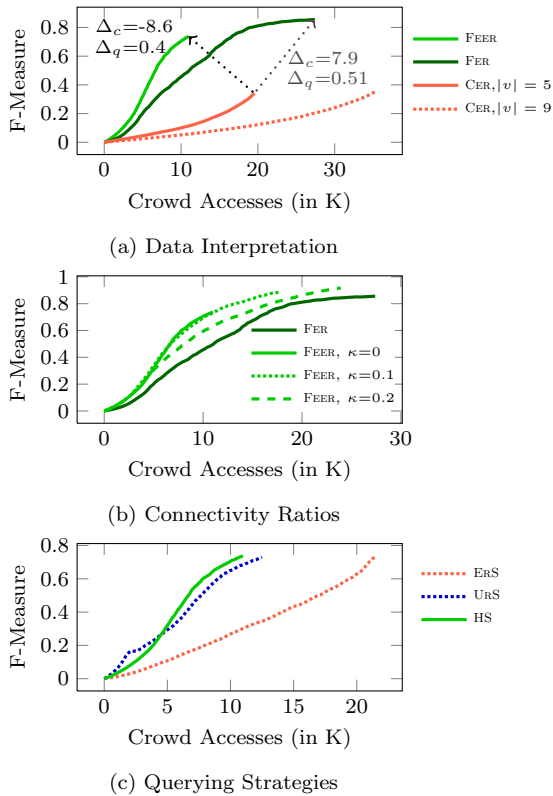

\smallsection{Performance of {\sc Fer} and {\sc Feer}}
As shown in \Autoref{fig:landmarksseqdata}, it is possible that {\sc Fer} outperforms {\sc Feer}.
While {\sc Feer} uses an adapting queuing system that obviously reduces the overall cost, {\sc Fer} requests candidate pairs as long as its queue is not empty.
Its monotonic queue is generated in the beginning and every candidate pair is polled from the queue in random order exactly once.
Given noisy answers from the crowd, this mechanism actually improves result quality intuitively because it is not aimed to minimize the cost but to ask every candidate pair at some point in time which is more costly but also more exhaustive than optimizing the queuing system. 
To verify this hypothesis, we implemented a modification of {\sc Feer} that allows us to vary the connectivity of the entities associated with each record $\kappa$ which varies the number of record pairs per entity combination.
A higher value of $\kappa$ resembles more record pairs that are requested to test the relationship of two entities.
For example, if the first entity contains four records and the second entity contains three records, there exist a total of 12 record pairs.
Instead of selecting one at random, it will select two random pairs if $\kappa$ is set to 0.2.
Increasing the connectivity of the entities has immediate consequences:
First, the cost of {\sc Feer} increases and second, the quality of {\sc Feer} improves significantly for this dataset as shown in \Autoref{fig:landmarksseqconnect}.
With $\kappa$ set to 0.2, we now observe a better quality to cost ratio for {\sc Feer} (result quality of 0.916 with a total cost of 23,918 cost units) than even {\sc Fer} can offer (result quality of 0.856 with a total cost of 27,197 cost units).

Finally, \Autoref{fig:landmarkstime} shows the performance comparison of our algorithms for this dataset in sequential mode without any modifications to the connectivity.
We observe that in this setup {\sc Fer} takes at most 8.12ms per update while the update cost of {\sc Feer} is 51.93ms.
The reason why {\sc Feer} is slower that {\sc Fer} lies in the readjustment of the queuing system:
As the dataset contains noise, record pairs get occasionally reinserted into the queue.
Generating these new record pairs and adjusting their position in the queue incurs computational overhead.
This overhead is then reflected in the execution time.

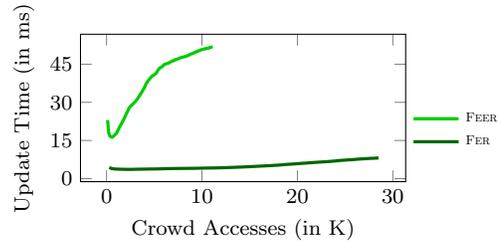
\begin{figure}[t]
\centering
\pgfplotstableread{graphdata/landmarks_time.dat}{\landmarkstime}
\begin{subfigure}[b]{.7\columnwidth}
\begin{tikzpicture}[font=\small]
	\small
	\begin{axis}
	[xlabel=Crowd Accesses (in K),
	ylabel= Update Time (in ms),
	legend style={at={(axis cs:32,10)},anchor=south west, nodes=right, draw=none, outer sep = 0pt, inner sep = 0pt, font=\tiny},
	xticklabel style={/pgf/number format/.cd,fixed,precision=2},
    xtick={0,10,20,30,40,50},
    ytick={0,15,30,45},
    width=\columnwidth,
    height=.6\columnwidth
    ]
	\addplot [green2,very thick, mark options=solid] table [x={cost_feer}, y={time_feer}] {\landmarkstime};
	\addlegendentry{{\sc Feer}}
	\addplot [green1,very thick, mark options=solid] table [x={cost_fer}, y={time_fer}] {\landmarkstime};
	\addlegendentry{{\sc Fer}}
	\end{axis}%
\end{tikzpicture}
\end{subfigure}
\caption{Time measurements for {\sc Fer} and {\sc Feer} in sequential mode for the landmarks dataset.}\label{fig:landmarkstime}
\end{figure}

\smallsection{Summary}
\Autoref{fig:landmarksseqdata} and \Autoref{fig:landmarksseqconnect} show that for this dataset {\sc Feer} reaches a better cost/quality trade-off faster than any other approach if the connectivity parameter $\kappa$ is adjusted.
Furthermore, all fault-tolerant strategies significantly outperform {\sc Cer} in terms of quality as we observe a minimal difference of at least 0.4 on the f-measure.
This improvement can be reached at lower cost than needed for {\sc Cer} for {\sc Feer}.
We also show that the low quality performance of {\sc Cer} does not depend on its available budget as the output quality only minimally increases with a higher budget (\Autoref{fig:landmarksseqdata}).

\begin{figure}[t]
\centering
\pgfplotstableread{graphdata/landmarks_approaches_parallel.dat}{\approachparallel}
\pgfplotstableread{graphdata/landmarks_strategycomparison_parallel.dat}{\strategyparallel}
\begin{subfigure}[b]{.7\columnwidth}
\begin{tikzpicture}[font=\small]
	\small
	\begin{axis}
	[xlabel=Crowd Accesses (in K),
	ylabel=F-Measure,
	legend style={at={(axis cs:43,0.8)},anchor=north west, nodes=right, draw=none, outer sep = 0pt, inner sep = 0pt, font=\tiny},
	xticklabel style={/pgf/number format/.cd,fixed,precision=2},
    xtick={0,10,20,30,40,50},
    ytick={0,0.2,0.4,0.6,0.8,1.0},
    width=\columnwidth,
    height=.6\columnwidth
    ]
	\addplot [green2,very thick, mark options=solid] table [x={cost_feer}, y={feer}] {\approachparallel};
	\addlegendentry{{\sc Feer}}
	\addplot [green2, dashed, very thick, mark options=solid] table [x={cost_feer02}, y={feer02}] {\approachparallel};
	\addlegendentry{{\sc Feer}, $\kappa=0.2$}
	\addplot [green1,very thick, mark options=solid] table [x={cost_fer}, y={fer}] {\approachparallel};
	\addlegendentry{{\sc Fer}}
	\addplot [red1, very thick, mark options=solid] table [x={cost_cer5}, y={cer5}] {\approachparallel};
	\addlegendentry{{\sc Cer}, $|v|=5$}
	\addplot [red1, densely dotted,very thick, mark options=solid] table [x={cost_cer9}, y={cer9}] {\approachparallel};
	\addlegendentry{{\sc Cer}, $|v|=9$}

	\draw[<-, black!70, thick, dotted] (axis cs:27.362,0.8624133374) -- (axis cs:19.493,0.3896961529);
	\draw[<-, black, thick, dotted] (axis cs:13.075,0.704404749) -- (axis cs:19.493,0.3896961529);
	
	\node[black!70] (d) at (axis cs: 27.6,.58) {\scriptsize$\Delta_c$=7.9};
	\node[black!70] (d) at (axis cs: 28.2,.43) {\scriptsize$\Delta_q$=0.47};
	\node[black] (d) at (axis cs: 10,.53) {\scriptsize$\Delta_c$=-6.4};
	\node[black] (d) at (axis cs: 10.2,.38) {\scriptsize$\Delta_q$=0.32};
	\end{axis}%
\end{tikzpicture}
\caption{Data Interpretation}\label{fig:landmarkspardata}
\end{subfigure}

\begin{subfigure}[b]{.7\columnwidth}
\begin{tikzpicture}[font=\small]
	\small
	\begin{axis}
	[xlabel=Crowd Accesses (in K),
	ylabel=F-Measure,
	legend style={at={(axis cs:20,.65)},anchor=north west, nodes=right, draw=none, outer sep = 0pt, inner sep = 0pt, font=\tiny},
	xticklabel style={/pgf/number format/.cd,fixed,precision=2},
    xtick={0,5,10,15,20,25,30,35,40,45,50},
    ytick={0,0.2,0.4,0.6,0.8,1.0},
    width=\columnwidth,
    height=.6\columnwidth
    ]
	\addplot [red1, densely dotted,very thick, mark options=solid] table [x={cost_cau}, y={cau}] {\strategyparallel};
	\addlegendentry{{\sc ErS}}
	\addplot [blue1, densely dotted,very thick, mark options=solid] table [x={cost_opt}, y={opt}] {\strategyparallel};
	\addlegendentry{{\sc UrS}}
	\addplot [green2, very thick, mark options=solid] table [x={cost_hyb}, y={hyb}] {\strategyparallel};
	\addlegendentry{{\sc HS}}
	\end{axis}%
\end{tikzpicture}
\caption{Querying Strategies}\label{fig:landmarksparqueue}
\end{subfigure}
\caption{Landmarks dataset, parallel execution mode with varying data interpretation models and querying strategies.}
\end{figure}
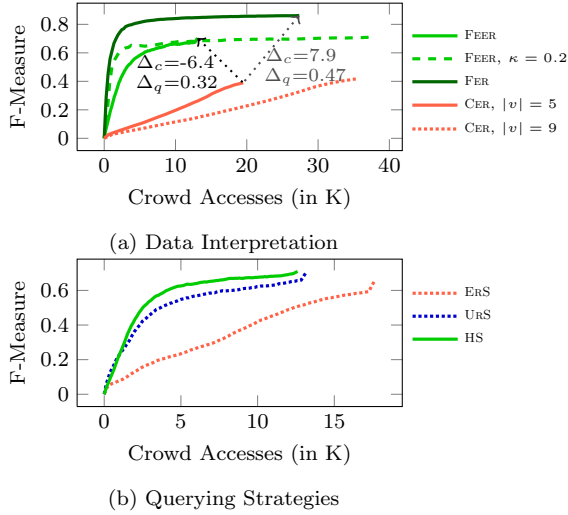

\subsubsection{Querying strategies}
When comparing the three proposed querying strategies for {\sc Feer} in \Autoref{fig:landmarksseqqueue}, similar cost/quality trade-offs in the real-world dataset than estimated with the synthetic crowd (\Autoref{fig:perfectcrowd}) can be observed:
Here, {\sc HS} clearly outperforms both alternative strategies as it provides better result quality at lower cost.
The exact improvement that {\sc HS} may offer depends on the underlying structure of the dataset and the crowd:
More noise decreases the performance of {\sc UrS} and highlights the robustness of {\sc HS}.
In comparison to {\sc ErS}, {\sc HS} provides better cost/quality trade-off if the number of entities in the result is large which highlights the exploitive nature of its processing of negative decision candidates.
In contrast to the slight cost increase in \Autoref{fig:perfectcrowd}, all queuing strategies are 40-60\% more expensive in the real-world experiment due to non-uniform error distribution.
More specifically, we observe for the landmarks dataset that 5.2\% of all pair-wise decisions are contested, i.e.,~there exist at least two crowd workers who have a different opinion than the other workers for the same pair of records.

\subsubsection{Parallelization}
Parallelizing the different data interpretation strategies (\Autoref{fig:landmarkspardata}) with intra and inter-cluster parallelization methods results in a worse quality/cost development for {\sc Feer} than {\sc Fer} though both are substantially better than {\sc Cer}.
The reason for the efficiency decrease of {\sc Feer} lies in the noisy information that parallelization induces:
It artificially creates race conditions, i.e.,~edges compete with each other for the highest score which results in them being marked as noisy and removed from the graph.
This behavior also explains the high cost of {\sc Feer} with $\kappa=0.2$:
It generates multiple connections per entity comparison thus further increasing the degree of parallelism and implicitly the noise.
Comparing these results to the results of the synthetic experiments previously obtained in \Autoref{fig:landmarkssynth01par}, we make similar observations for {\sc ErS} and {\sc UrS} than before while {\sc HS} increases its advantage over both strategies due to the non-uniform noise distribution in this dataset.

\begin{figure}[t]
\centering
\pgfplotstableread{graphdata/publications_approaches_sequential.dat}{\approachsequential}
\pgfplotstableread{graphdata/publications_approaches_sequential_s0.dat}{\approachsequentialszero}
\begin{subfigure}[b]{.7\columnwidth}
\begin{tikzpicture}[font=\small]
	\small
	\begin{axis}
	[xlabel=Crowd Accesses (in K),
	ylabel=F-Measure,
	legend style={at={(axis cs:52,.75)},anchor=north west, nodes=right, draw=none, outer sep = 0pt, inner sep = 0pt, font=\tiny},
	xticklabel style={/pgf/number format/.cd,fixed,precision=2},
    xtick={0,10,20,30,40,50,60},
    ytick={0,0.2,0.4,0.6,0.8,1.0},
    width=\columnwidth,
    height=.6\columnwidth
    ]
	\addplot [green2,very thick, mark options=solid] table [x={cost_feer}, y={feer}] {\approachsequentialszero};
	\addlegendentry{{\sc Feer}}
	\addplot [green1,very thick, mark options=solid] table [x={cost_fer}, y={fer}] {\approachsequentialszero};
	\addlegendentry{{\sc Fer}}
	\addplot [red1, very thick, mark options=solid] table [x={cost_cer5}, y={cer5}] {\approachsequentialszero};
	\addlegendentry{{\sc Cer}, $|v|=5$}
	\addplot [red1, densely dotted,very thick, mark options=solid] table [x={cost_cer9}, y={cer9}] {\approachsequentialszero};
	\addlegendentry{{\sc Cer}, $|v|=9$}
	
	\draw[<-, black, thick, dotted] (axis cs:15.077,0.7980440001) -- (axis cs:24.892,0.4401726128);
	\draw[<-, black!70, thick, dotted] (axis cs:14.202,0.7849376964) -- (axis cs:24.892,0.4401726128);
	
	\node[black] (d) at (axis cs: 27.5,.72) {\scriptsize$\Delta_c$=-9.8};
	\node[black] (d) at (axis cs: 27.8,.59) {\scriptsize$\Delta_q$=0.36};
	\node[black!70] (d) at (axis cs: 9,.65) {\scriptsize$\Delta_c$=-10.7};
	\node[black!70] (d) at (axis cs: 8.7,.52) {\scriptsize$\Delta_q$=0.35};
	
	\end{axis}%
\end{tikzpicture}
\caption{$0<s<1$}\label{fig:publicationssequapps0}
\end{subfigure}

\begin{subfigure}[b]{.7\columnwidth}
\begin{tikzpicture}[font=\small]
	\small
	\begin{axis}
	[xlabel=Crowd Accesses (in K),
	ylabel=F-Measure,
	legend style={at={(axis cs:2.5,.7)},anchor=north west, nodes=right, draw=none, outer sep = 0pt, inner sep = 0pt, font=\tiny},
	xticklabel style={/pgf/number format/.cd,fixed,precision=2},
    xtick={0,1,2,3,4,5,6},
    ytick={0,0.2,0.4,0.6,0.8,1.0},
    width=\columnwidth,
    height=.6\columnwidth
    ]
	\addplot [green2,very thick, mark options=solid] table [x={cost_feer}, y={feer}] {\approachsequential};
	\addlegendentry{{\sc Feer}}
	\addplot [green1,very thick, mark options=solid] table [x={cost_fer}, y={fer}] {\approachsequential};
	\addlegendentry{{\sc Fer}}
	\addplot [red1, very thick, mark options=solid] table [x={cost_cer5}, y={cer5}] {\approachsequential};
	\addlegendentry{{\sc Cer}, $|v|=5$}
	\addplot [red1, densely dotted,very thick, mark options=solid] table [x={cost_cer9}, y={cer9}] {\approachsequential};
	\addlegendentry{{\sc Cer}, $|v|=9$}
	
	\draw[<-, black!70, thick, dotted] (axis cs:1.637,0.7807331504) -- (axis cs:1.213,0.4432895355);
	\draw[<-, black, thick, dotted] (axis cs:0.947,0.7510319664) -- (axis cs:1.213,0.4432895355);
	
	\node[black!70] (d) at (axis cs: 1.85,.68) {\scriptsize$\Delta_c$=0.4};
	\node[black!70] (d) at (axis cs: 1.89,.53) {\scriptsize$\Delta_q$=0.34};
	\node[black] (d) at (axis cs: 1.2,.3) {\scriptsize$\Delta_c$=-0.3};
	\node[black] (d) at (axis cs: 1.21,.15) {\scriptsize$\Delta_q$=0.31};
	
	\end{axis}%
\end{tikzpicture}
\caption{$0.3<s<1$}\label{fig:publicationssequapp}
\end{subfigure}
\caption{Publications dataset, sequ. execution mode with varying data interpretation models and similarity thresholds.}
\end{figure}
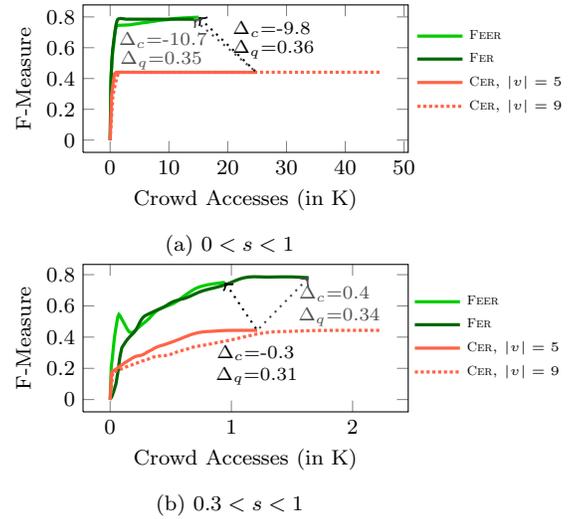

\subsection{Publications Results}\label{subsec:publications}
In contrast to the landmarks dataset, the publications dataset consists of records with string-based attributes which allow for an easy application of string similarity metrics to prune the search space and reduce information access cost.
That is, in this set of experiments, we first apply automated ER to obtain those candidate pairs that are uncertain which are then verified through crowdsourcing.
In that context, the similarity metric $s$ determines the level of matching with the Jaccard similarity as follows:
First, if $s$ exceeds an upper threshold, a positive decision with the maximum number of crowd votes is added to the votes graph.
Record pairs that are below the lower threshold of $s$ trigger a negative decision in the same manner.
Furthermore, $s$ is used to initiate all queues according to the current similarity belief.

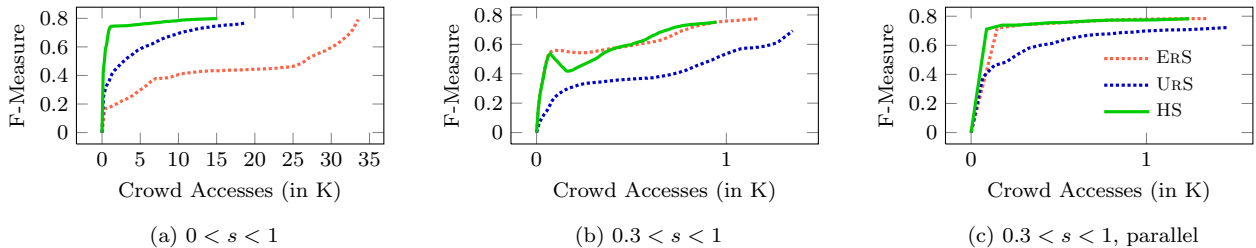
\begin{figure*}[t]
\centering
\pgfplotstableread{graphdata/publications_strategycomparison_sequential.dat}{\strategysequential}
\pgfplotstableread{graphdata/publications_strategycomparison_sequential_s0.dat}{\strategysequentialszero}
\pgfplotstableread{graphdata/publications_strategycomparison_parallel.dat}{\strategyparallel}
\begin{subfigure}[b]{.32\textwidth}
\begin{tikzpicture}[font=\small]
\definecolor{green1}{RGB}{0, 102, 0}
\definecolor{green2}{RGB}{0, 204, 0}
\definecolor{red1}{RGB}{255, 99, 71}
	\small
	\begin{axis}
	[xlabel=Crowd Accesses (in K),
	ylabel=F-Measure,
	xticklabel style={/pgf/number format/.cd,fixed,precision=2},
    xtick={0,5,10,15,20,25,30,35,40,45,50},
    ytick={0,0.2,0.4,0.6,0.8,1.0},
    width=\columnwidth,
    height=.6\columnwidth
    ]
	\addplot [red1, densely dotted,very thick, mark options=solid] table [x={cost_cau}, y={cau}] {\strategysequentialszero};
	\addplot [blue1, densely dotted,very thick, mark options=solid] table [x={cost_opt}, y={opt}] {\strategysequentialszero};
	\addplot [green2,very thick, mark options=solid] table [x={cost_hyb}, y={hyb}] {\strategysequentialszero};
	\end{axis}%
\end{tikzpicture}
\caption{$0<s<1$}\label{fig:publicationssequqss0}
\end{subfigure}
\begin{subfigure}[b]{.32\textwidth}
\begin{tikzpicture}[font=\small]
	\small
	\begin{axis}
	[xlabel=Crowd Accesses (in K),
	ylabel=F-Measure,
	xticklabel style={/pgf/number format/.cd,fixed,precision=2},
    xtick={0,1,2,3,4,5,6},
    ytick={0,0.2,0.4,0.6,0.8,1.0},
    width=\columnwidth,
    height=.6\columnwidth
    ]
	\addplot [red1, densely dotted,very thick, mark options=solid] table [x={cost_cau}, y={cau}] {\strategysequential};
	\addplot [blue1, densely dotted,very thick, mark options=solid] table [x={cost_opt}, y={opt}] {\strategysequential};
	\addplot [green2,very thick, mark options=solid] table [x={cost_hyb}, y={hyb}] {\strategysequential};
	\end{axis}%
\end{tikzpicture}
\caption{$0.3<s<1$}\label{fig:publicationssequqs}
\end{subfigure}
\begin{subfigure}[b]{.32\textwidth}
\begin{tikzpicture}[font=\small]
	\small
	\begin{axis}
	[xlabel=Crowd Accesses (in K),
	ylabel=F-Measure,
	legend style={at={(axis cs:.75,.6)},anchor=north west, nodes=right, draw=none, outer sep = 0pt, inner sep = 0pt, font=\scriptsize},
	xticklabel style={/pgf/number format/.cd,fixed,precision=2},
    xtick={0,1,2,3,4,5,6},
    ytick={0,0.2,0.4,0.6,0.8,1.0},
    width=\columnwidth,
    height=.6\columnwidth
    ]
	\addplot [red1, densely dotted,very thick, mark options=solid] table [x={cost_cau}, y={cau}] {\strategyparallel};
	\addlegendentry{{\sc ErS}}
	\addplot [blue1, densely dotted,very thick, mark options=solid] table [x={cost_opt}, y={opt}] {\strategyparallel};
	\addlegendentry{{\sc UrS}}
	\addplot [green2,very thick, mark options=solid] table [x={cost_hyb}, y={hyb}] {\strategyparallel};
	\addlegendentry{{\sc HS}}
	\end{axis}%
\end{tikzpicture}
\caption{$0.3<s<1$, parallel}\label{fig:publicationsparqs}
\end{subfigure}
\caption{Publications dataset with varying execution modes, querying strategies, and similarity thresholds.}
\end{figure*}

\subsubsection{Data interpretation}
We make three observations when comparing the different data interpretation methodologies (\Autoref{fig:publicationssequapp} and \Autoref{fig:publicationssequapps0}).
First, fault-tolerant mechanisms clearly outperform consensus-based mechanisms regardless of the similarity threshold $s$ that is applied to the record pairs.
Independent of the number of votes spent by {\sc Cer}, we observe a maximum precision of 0.92 and a maximum recall of 0.29, indicating that, again, {\sc Cer} keeps records in different entities that should belong to the same entity.
The second observation is that spent budget can be signifantly reduced if automatic similarity measures are used before applying the data interpretation mechanisms but it does not influence the quality outcome of either technique.
This suggests that the decisions made by the similarity metric are in fact decisions that the crowd workers make as well.
As a result, tightening the thresholds only prunes undisputed decisions and thus reduces crowdsourcing costs without modifying the result quality.
Third, we observe that for all approaches the cost/quality gain first increases and then stagnates, which is especially evident in \Autoref{fig:publicationssequapps0}.
The average entity size in the publications dataset is only 3.47 which results in a large number of negative decisions between entity pairs overall.
As cost is not bound here, all undetermined pairs have to be resolved until the algorithms are finished.
The significant number of negative decisions at the end of the execution (which is visibly more skewed than in the previously examined landmarks dataset) is due to the application of the similarity metric a priori, which boosts decisions that are more likely to be positive to the front of the priority queue.

\smallsection{Entity connectivity}
The concept of connectivity ratios that was introduced for the landmarks dataset can also be applied to the publications dataset.
Similarily to the previous experiment, we observe an increase in cost and quality when $\kappa$ is increased.
The changes to the resulting ER are not as signficant as observed for the landmarks dataset but are within 2\% of its quality while the cost increases by approximately 10\%.

\smallsection{Summary}
With these experiments, we show that fault-tolerant ER and automatic similarity metrics can be tightly integrated and result in a good quality solution for a smaller overall budget.
If the similarity metric suits the dataset, it can efficiently decrease the search space for crowdsourced ER without a loss in quality and an improvement in cost (here 90\% of the budget).
Furthermore, these experiments show that reducing the search space does not necessarily benefit the decision function, as the output quality is still dependent on how noisy votes are handled.
Here, both {\sc Feer} and {\sc Fer} significantly outperform {\sc Cer}.

\subsubsection{Querying strategies}
With a prepruned search space, querying strategies behave differently than in an unbiased decision space.
We observe that for an increased lower bound of $s$, {\sc ErS} performs better than {\sc UrS} (\Autoref{fig:publicationssequqs}) while {\sc UrS} dominates {\sc ErS} in an uninformed setup (\Autoref{fig:publicationssequqss0}).
Obviously, {\sc ErS} performs well if the candidate space is comprised of mostly positive candidate pairs.
On the other hand, it incurs higher cost if it likely asks negative decisions which occur more often if the lower threshold for $s$ is decreased.
The development of both {\sc HS} and {\sc UrS} are more straightforward as they rely on an exploratory way of evaluating the search space which makes them less dependent on prefiltering.

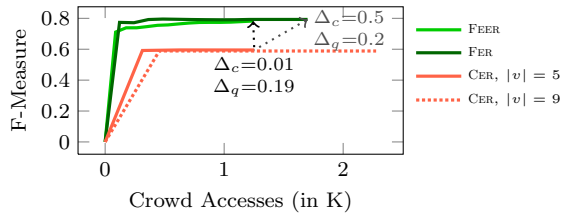
\begin{figure}[t]
\centering
\pgfplotstableread{graphdata/publications_approaches_parallel.dat}{\approachparallel}
\begin{subfigure}[b]{.7\columnwidth}
\begin{tikzpicture}[font=\small]
	\small
	\begin{axis}
	[xlabel=Crowd Accesses (in K),
	ylabel=F-Measure,
	legend style={at={(axis cs:2.6,.8)},anchor=north west, nodes=right, draw=none, outer sep = 0pt, inner sep = 0pt, font=\tiny},
	xticklabel style={/pgf/number format/.cd,fixed,precision=2},
    xtick={0,1,2,3,4,5,6},
    ytick={0,0.2,0.4,0.6,0.8,1.0},
    width=\columnwidth,
    height=.6\columnwidth
    ]
	\addplot [green2,very thick, mark options=solid] table [x={cost_feer}, y={feer}] {\approachparallel};
	\addlegendentry{{\sc Feer}}
	\addplot [green1,very thick, mark options=solid] table [x={cost_fer}, y={fer}] {\approachparallel};
	\addlegendentry{{\sc Fer}}
	\addplot [red1, very thick, mark options=solid] table [x={cost_cer5}, y={cer5}] {\approachparallel};
	\addlegendentry{{\sc Cer}, $|v|=5$}
	\addplot [red1, densely dotted,very thick, mark options=solid] table [x={cost_cer9}, y={cer9}] {\approachparallel};
	\addlegendentry{{\sc Cer}, $|v|=9$}
	
	\draw[<-, black!70, thick, dotted] (axis cs:1.701,0.7928202442) -- (axis cs:1.253,0.5950340054);
	\draw[<-, black, thick, dotted] (axis cs:1.244,0.7824859491) -- (axis cs:1.253,0.5950340054);
	
	\node[black!70] (d) at (axis cs: 2.05,.8) {\scriptsize$\Delta_c$=0.5};
	\node[black!70] (d) at (axis cs: 2.05,.65) {\scriptsize$\Delta_q$=0.2};
	\node[black] (d) at (axis cs: 1.25,.5) {\scriptsize$\Delta_c$=0.01};
	\node[black] (d) at (axis cs: 1.25,.35) {\scriptsize$\Delta_q$=0.19};

	\end{axis}%
\end{tikzpicture}
\end{subfigure}
\caption{Publications dataset ($0.3<s<1$), in parallel execution mode with varying data interpretation models.}\label{fig:publicationsparapp}
\end{figure}

\subsubsection{Parallelization}
Surprisingly, parallelizing the candidate pairs for the publications dataset has significant impact on the performance of {\sc Cer} (\Autoref{fig:publicationsparapp}).
Here, the observed f-measure changes because the recall of both {\sc Cer} variations increases to a maximum of .44 compared to the .29 achieved during sequential execution.
The reason for this behavior attests to the instability of consensus-based approaches:
It can be found in the order in which the record pairs that are extracted from the priority queue.
For this dataset specifically, the ordering of the candidate pairs removes a candidate pair from the initial batch of tasks that has been falsely identified by the crowd and as a result, the output quality improves.
In contrast, both {\sc Fer} and {\sc Feer} have consistent quality when executed sequentially or in parallel.
\Autoref{fig:publicationsparqs} shows the performance of the different querying strategies in parallel execution mode.
Similarily to \Autoref{fig:publicationssequqs}, we note that {\sc ErS} is a valid alternative queuing strategy to {\sc HS} given a prepruned search space, even if its cost is slightly higher than for {\sc HS}.

\subsection{Discussion}
Fault-tolerance is a requirement for entity resolution when handling unreliable data sources such as the crowd as shown in this section for two different datasets.
Noisy information has direct impact on result quality  and cannot be recovered from if the applied ER mechanism is not aware of these imperfections.
Furthermore, we show that choices concerning task ordering are essential to the success of any ER mechanism:
Queuing strategies as well as parallelization techniques impact the result quality and cost and cause different results in a noisy as in a perfect execution environment.
Specifically, we observe that parallelization may improve execution time, i.e.,~the end-to-end time spent on the ER process, but it increases the allocated budget and often does not achieve the same ER quality level as a sequential execution could provide.
The relative success of the presented queuing strategies then is dependent on the dataset itself, the size of the entities it contains, and also whether pre-pruning is an option.
We have shown that a hybrid queuing strategy is a robust mechanism to order tasks consciously of the current state of the votes graph as alternative to pure error or uncertainty reduction strategies.

\smallsection{Scalability}
For the sake of completeness, these experiments have been executed on datasets that were collected through a crowdsourcing platform.
As a result, these datasets are limited in size though experiments with synthetically generated larger datasets show the same tendencies for all algorithms and strategies (see \Autoref{subsec:crowderror} for an example of how synthetic experiments are conducted).
We argue that as shown in the above real-world experiments, the state-of-the-art mechanisms incur comparable cost in terms of crowd accesses.
Thus, fault-tolerant mechanisms on average provide higher quality results at the same cost because they are able to prioritize and question important record pairs that are central to the entity resolution solution.
For larger datasets than examined in this set of experiments, we can imagine techniques such as automatic similarity metrics to minimize the candidate space for record pairs similar to those applied on the publications dataset here.
Note that this does not diminish the impact of fault-tolerant decision mechanisms on the output quality as shown in our experiments (\Autoref{subsec:publications}).

In terms of computational performance, our experiments show that the entity size as well as noise level in the answer set are correlated to the update propagation performance of {\sc update} (Algorithm~\ref{alg:recUpdate}).
That is, large clusters cause a large amount of positive decisions in the votes graph which in return means more update propagation as positive edges are always traversed independent of whether the candidate path is positive or negative.
Furthermore, noisy edges obviously require more computation because every updated edge needs to be propagated into the votes graph.
Here, we argue that computational scalability is often not a problem for the execution engine because
\begin{itemize}
  \item the crowd is slower than the time taken to update the votes graph especially if it is constructed as a chain similar to what all presented approaches here try to achieve by leveraging transitivity.
  \item parallelization as described previously allows a variety of workers to respond to several tasks at the same time thus decreasing the end-to-end time spent on an ER problem.
\end{itemize}

\smallsection{(Crowdsourced) Entity Resolution}
Automatic ER algorithms and their corresponding approximations, \cite{DBLP:journals/ml/BansalBC04,DBLP:journals/pvldb/HassanzadehCML09}, are useful alternatives to crowdsourced ER algorithms in a variety of use cases:
For example for the Cora dataset which contains text-based content, string-based similarity metrics have been shown to provide high quality output.
In this paper, the purpose of using this dataset is to be comparable to previously done research in the same area.
The use cases that we target with crowdsourced ER for real-world use cases are exemplified better through the landmarks dataset where pair-wise similarities are not straightforward to compute.
In fact, identification of objects in pictures is still a hard task for computers, \cite{DBLP:journals/corr/WeyandL14}.
In these cases, crowdsourced ER can be used to enhance and complement object identification either as a standalone solution or in collaboration with automatic similarity measurements that have been developed for visual computing.
Analogous, we imagine crowdsourced ER to be used in other domains where information is not of the same data type or cannot be well correlated with automatic measures.

\section{Related Work}\label{sec:related}

\smallsection{Entity Resolution}
Entity resolution (also known as entity reconciliation, duplicate detection, or record linkage) is a critical task for data integration and cleaning. 
It has been studied extensively for several decades (see \cite{DBLP:journals/tkde/ElmagarmidIV07} for a survey). 
There are a variety of approaches to ER ranging from local decision functions such as transitive closure, \cite{DBLP:journals/pvldb/HassanzadehCML09}, to global objective functions such as cut or correlation clustering, \cite{DBLP:journals/ml/BansalBC04}.
In this work, we combine the idea of having cohesive clusters of records with the quality guarantees that local decisions can provide, i.e.,~if the crowd argues for one decision over another, our algorithms ensure consistency with that decision.
For automatic ER algorithms, cost optimization is often not as essential as in crowdsourced ER because it is not directly correlated with monetary cost.
Nevertheless, approaches such as progressive ER, \cite{DBLP:journals/pvldb/AltowimKM14}, and incremental ER, \cite{DBLP:journals/pvldb/GruenheidDS14}, work under the assumption that the information gain per update or new element to resolve should be maximized.
This idea is similar to what we want to achieve with the presented queuing strategies, although our mechanisms differ in that the votes of the crowd are iteratively collected and decisions may change over time, i.e.,~they become invalidated, which is not the case for automatic ER approaches.
Another interesting line of work if adjusted to the crowdsourcing context is scalable entity resolution \cite{DBLP:conf/kdd/GetoorM13,DBLP:journals/pvldb/RastogiDG11}.
Here, the premise is to break down large entity resolution problems into smaller parts to enhance performance.
It is well applicable in the context of crowdsourced entity resolution although it still assumes perfect knowledge about record relationships which is different than from the underlying assumptions of fault-tolerant entity resolution.

\smallsection{Crowdsourced Entity Resolution}
Recently, hybrid human/machine entity resolution algorithms have attempted to automatically integrate humans as part of the ER process to increase the reconciliation quality \cite{lee2013hybrid,DBLP:conf/nips/YiJJJY12}.
In \cite{Wang:2012} the authors combine automatic machine learning techniques with crowdsourcing, whereas \cite{DBLP:conf/sigmod/WangLKFF13} extends the work to further reduce the cost of crowdsourcing by taking transitive relationships into account and also adjusting the crowdsourcing process according to automated similarity measurements.
One of the assumptions that is commonly made in previous work is to consider worker quality as an orthogonal problem as well as not taking negative feedback from humans (i.e.,~that two entities do not match) into account when creating an ER solution \cite{DBLP:conf/sigmod/WangLKFF13,DBLP:journals/pvldb/WhangLG13}. 
Thus, eventual conflicts in the crowdsourced comparisons are discarded or ignored.
In contrast, our approach uses all available information without filtering any of the retrieved data even if contradicting.
Additionally, we closely examine queuing strategies that make a decision based on all available information rather than using a subset of information to determine the next record pair that is requested from the crowd \cite{DBLP:journals/pvldb/VesdapuntBD14}.
The strategies that we evaluate in this context allow users to determine whether they want to put emphasis on high precision or high recall in their result set to allow for flexibility when constructing the ER solution.
Furthermore, there has been work on probabilistic crowdsourced ER \cite{ilprints1097} which proposes a maximum likelihood approach but the proposed strategy is NP-hard and thus infeasible to compute in an online setting (see \Autoref{subsec:strategies} for details).
The techniques shown in this work are specifically designed for efficient online computation.
Additionally, we provide quality guarantees and explore different priorization techniques for record pairs which has not been discussed in any of the related crowdsourced ER works.

\smallsection{Crowd Quality}
Since defining and computing good similarity measures is not always possible, there has been further work \cite{Bellare:2012:ASE:2339530.2339707,DBLP:conf/sigmod/Gokhale2014,DBLP:conf/nips/GomesWKP11,yi_crowdclustering_2012} that minimize the use of distance functions for record comparisons.
These methods either rely on Bayesian modelling or similar mechanisms to approximate the answers of their data sources after extensive data collection or apply knowledge specific to a certain platform and its characteristics \cite{Bozzon:2013:RC:2488388.2488403,Ipeirotis:2010:QMA:1837885.1837906,Kosinski:2012:CIM:2380718.2380739}.
Error intervals for crowd workers or error estimates per worker group are alternative ways to model worker quality \cite{DBLP:conf/kdd/JoglekarGP13}.
In contrast, our method provides good ER results in the absence of similarity distances and is able to provide an ER result without any prior training of our mechanism. 
Information on the quality of crowd workers can be leveraged with our approach by requiring high quality workers, for example determined through their behavioral patterns \cite{Kazai:2011:WTP:2063576.2063860}, to answer the current top candidate pair.
Thus, this type of research can be used to enhance the computed ER result.
Nevertheless, noise in the answer set cannot be excluded categorically even if worker quality is high as they may make mistakes and provide erroneous information.

\smallsection{Crowdsourced Database Operators}
There has been a lot of research on crowdsourced operators (filtering, top-k, and entity resolution) under the assumption of predefined error behavior of the crowd workers.
This research can be divided into two categories:
Approaches that rely on the crowd to give answers that can be monotonously aggregated \cite{Wang:2012,Whang:2012} and those techniques that take specific error behavior of the crowd into consideration \cite{Davidson:2013:UCT:2448496.2448524,DBLP:conf/sigmod/GuoPG12,Parameswaran:2012:CAF:2213836.2213878}.
Our approaches vary from the first group of algorithms as we tolerate and in fact embrace noisy behavior of the crowd.
To the best of our knowledge, none of the techniques developed in the second group of algorithms can be used to resolve the entity resolution problem efficiently.

\smallsection{Performance Optimizations}
Orthogonal, as it can be easily employed on top of our fault-tolerant entity resolution framework, but relevant to our results is research that focuses on finding methods to improve the information gain from the crowd.
That is, instead of looking at binary record comparisons as we do here, there have been alternative methods to batch tasks or to define interfaces specifically for certain crowdsourcing tasks that enable crowd workers to be more precise in their answers \cite{komarov13:crowdsourcing} and to convey more information in a single task \cite{DBLP:journals/pvldb/0002KMMO12}.
Additionally, machine learning has been used to characterize the crowd as well as to optimize budget spending patterns \cite{DBLP:journals/ior/KargerOS14,Lease11}.
These spending patterns are similar to the queuing strategies that we present in this work although the idea in this type of work is to determine and leverage these patterns per worker.
In contrast, we make more general observations on the applicability of recall resp. precision-oriented queuing strategies and how varying these strategies modifies output quality.

\section{Conclusion}{\label{sec:conclusion}}

In this work, we discussed the problem of entity resolution with unreliable data collected through crowdsourcing.
To handle this type of noisy data, we defined fault-tolerant mechanisms that interpret the pair-wise decisions on record relationships made by the crowd.
We then discussed how these mechanisms can be integrated into an incremental ER framework that computes a consistent ER solution on the fly.
In the second part of this work, we showed how uncertainty affects the next-crowdsource problem in terms of result quality and budget allocation.
In that context, we evaluated different queuing strategies that order tasks according to different objectives for example minimizing the error or uncertainty first.
Furthermore, intra- and inter-task parallelization mechanisms were added as execution modes to all techniques in order to evaluate the impact of task parallelization on the result quality and cost of the ER solutions.
The need for fault-tolerant mechanisms was clearly shown in the experimental evaluation where the devised fault-tolerant mechanisms outperform consensus-based strategies significantly for two real-world datasets.
Additionally, the evaluation clearly shows the impact of task ordering and parallelization mechanisms on crowdsourced entity resolution.
Note that the techniques presented in this paper are not restricted to being used in the context of crowdsourcing only but can be used for any unreliable data source.

\balance

\bibliographystyle{abbrv}
\bibliography{library}

\end{document}